\documentclass[runningheads]{llncs}
\usepackage[T1]{fontenc}

\usepackage{amsmath,stmaryrd,amssymb}
\usepackage{bbm,cmll}
\usepackage{enumitem}
\usepackage[inference]{semantic} 
\usepackage{proof} 
\usepackage{xifthen}
\usepackage{textcomp}

\usepackage{hyperref}
\usepackage{color}

\usepackage{adjustbox}

\providecommand*{\ifempty}[3]{\ifthenelse{\isempty{#1}}{#2}{#3}}

\newcommand{\eg}{\textit{e.g.}}
\newcommand{\ie}{\textit{i.e.}}
\newcommand{\cf}{\textit{cf.}}

\newcommand{\e}{\varepsilon}

\newcommand{\M}{\mathcal{M}}

\newcommand{\FF}{\mathfrak{F}}
\newcommand{\PP}{\mathfrak{P}}

\newcommand{\X}{\mathcal{X}}

\newcommand{\lol}{\multimap}
\newcommand{\lollol}{\multimapboth}
\newcommand{\ot}{\oplus}

\newcommand{\naturals}{\mathbb{N}}

\newcommand{\reals}{\mathbb{R}}
\newcommand{\preals}{[0,\infty)}

\newcommand{\extreals}{[0,\infty]}

\newcommand{\Prop}[1][]{\X_{#1}}

\newcommand{\al}{\mathbb{AL}}
\newcommand{\pl}{\mathbb{PL}}

\newcommand{\coloneq}{\mathrel{\mathop:}=}
\newcommand{\coloneqq}{\mathrel{\mathop{::}}=}
\makeatletter
\newcommand{\dotdiv}{\mathbin{\text{\@dotminus}}}

\newcommand{\sem}[1]{\llbracket #1 \rrbracket}

\newcommand{\@dotminus}{%
  \ooalign{\hidewidth\raise1ex\hbox{.}\hidewidth\cr$\m@th-$\cr}%
}
\makeatother

\newcommand{\infrule}[3][]{\infer[{\ifempty{#1}{}{(\textsc{#1})\;}}]{#3}{#2}}
\newcommand{\doubleinfrule}[3][]{\infer=[{\ifempty{#1}{}{(\textsc{#1})\;}}]{#3}{#2}}

\begin{document}
\title{Polynomial Lawvere Logic}
%
%
\author{Bacci Giorgio\inst{1} \and
Radu Mardare\inst{2} \and
Prakash Panangaden\inst{3} \and
Gordon Plotkin\inst{4}}
\authorrunning{G. Bacci et al.}
%
\institute{Dept.\ of Computer Science, Aalborg University, Denmark \and
School of Mathematical and Computer Sciences, Heriot-Watt University, Scotland \and
School of Computer Science, McGill University, Canada \and 
LFCS, School of Informatics, University of Edinburgh, Scotland }
\maketitle              
\begin{abstract}
We study Polynomial Lawvere logic ($\pl$), a logic defined 
over the Lawvere quantale of extended positive reals with sum as tensor, to which we add multiplication, thereby obtaining a semiring structure. $\pl$ is designed for complex quantitative reasoning, allowing judgements that express inequalities between polynomials on the extended positive reals. We introduce a deduction system and demonstrate its expressiveness by deriving a classical result from probability theory relating the Kantorovich and the total variation distances. Although the deductive system is not complete in general, we achieve completeness for finitely axiomatizable theories. The proof of completeness relies on the Krivine-Stengle Positivstellensatz (a variant of Hilbert's Nullstellensatz).

Additionally, we provide new complexity results, both for $\pl$ and its affine fragment $\al$, regarding two decision problems: satisfiability of a set of judgements and semantical consequence from a set of judgements. The former is NP-complete in $\al$ and in PSPACE for $\pl$; the latter is co-NP complete in $\al$ and in PSPACE for $\pl$.

\end{abstract}
\section{Introduction}
Recent developments in theoretical computer science have questioned the relevance of equality in semantics, advocating for more nuanced, quantitative approaches to equivalence. For instance, exact equality is often too rigid for probabilistic systems where small changes can disrupt equivalence between processes. To address this, researchers started to use metrics to measure differences, thus shifting the focus from strict equivalence to a quantitative comparison. This metric-based reasoning 
has also been applied to areas like privacy, security~\cite{AmorimGHKC17,ReedP10}, and computational resource analysis~\cite{DalLagoG21,LagoG22a}, and symbolic computation~\cite{GavazzoF23}.

As a result, theories of semantic equality have evolved into quantitative frameworks, focusing on measuring differences rather than asserting equality. Notable examples include theories for  program analyis~\cite{ArnoldN80,CrubilleL15,CrubilleL17,LagoGY19,LagoG21,LagoG22a}, distances for processes~\cite{Desharnais99b,Desharnais04,Ferns04,Ferns05,BacciBLM19,BacciBLMTB19,BaldanBKK14}, and quantitative equational logics over algebras of terms~\cite{Mardare16,Mardare17,Bacci18,Mardare21,Bacci21,MioSV21,MioSV22,Adamek22,AdamekDV23}. The latter, in particular, focuses on providing foundations for quantitative reasoning. The basic idea is to replace 
traditional equations $s = t$ between terms $s,t$ of an algebra with \emph{quantitative equations} of the form $s =_\e t$, expressing that $s$ and $t$ are at most $\e$ apart, for some real  $\e \geq 0$. Thus, quantitative algebraic theories are used to reason about the \emph{distances} between elements of an algebra.
However, equational logic is only one of many forms of logic and the question about how extensions of classical logic can be used to provide foundations for such quantitative reasoning remained open. 

In his seminal work~\cite{Lawvere73}, Lawvere proposed the idea that generalized metric spaces could be viewed as categories enriched over the quantale $\extreals$, the complete lattice of extended positive reals with the order reversed and sum as tensor. This naturally leads to considering logics where truth values are derived from the Lawvere quantale. In such a framework, a quantitative equation $s =_\e t$ is expressed as a judgement $\e \vdash s = t$, which corresponds to the inequality $\e \geq {}$``$s = t$''. Here, $s = t$ is a predicate valued in $\extreals$ (with Lawvere's generalized metric spaces being simply $\extreals$-valued preorders).

Building on Lawvere's concept, in~\cite{Bacci2023} we began exploring a class of quantitative logics, which we refer to as \textit{Lawvere logics}%
\footnote{The logics are named in honor of Lawvere.}. 
Among them, Affine Lawvere propositional logic ($\al$) was the most expressive. This logic features a tensor operation $\ot$ interpreted as addition (in $\extreals$), a linear implication $\lol$ (acting as the adjoint residuum of $\ot$), constants $r$ for all non-negative real numbers, and scalar multiplication by non-negative reals, allowing it to express all affine functions on $\extreals$. Logical conjunction and disjunction are derived operators.
Judgements in $\al$ are interpreted as affine inequalities on $\extreals$. A key innovation of \cite{Bacci2023} was the use of 
linear algebra results, specifically Farkas' Lemma~\cite{Farkas1902} and Motzkin's transposition theorem~\cite{Motzkin51}, to establish  completeness. This provided a significant connection between logical reasoning and classical results from arithmetic.
However, many real-world quantitative phenomena involve non-linear interactions, making it desirable to express polynomial inequalities. 


In this paper, we take on the challenge of developing \emph{Polynomial Lawvere Logic} ($\pl$). This logic extends $\al$
by adding multiplication as a new logical connective,  enabling the encoding of polynomial functions on $\extreals$, 
with judgements representing polynomial inequalities.
Our approach builds on Lawvere's idea by giving logical status to both sum (the tensor) and multiplication over $[0,\infty]$, with the key innovation being that the truth values now come from a semiring structure involving the quantale. 

Our \emph{main contributions} are:
\begin{enumerate}[noitemsep,wide,topsep=0pt]
\item 
We propose a deduction system for PL (Table~\ref{tab:rules}) and demonstrate its expressiveness by (a) deriving a classical result from probability theory relating the Kantorovich and the total variation distances and (b) giving an embedding of quantitative equational logic in $\pl$ (\cf\ Section~\ref{sec:applications}).

\item We prove (Theorem~\ref{th:incompleteness}) $\pl$ is incomplete in general, but is (Theorem~\ref{thm:completeness}) complete for finitely axiomatizable theories. The core of the completeness proof   differs significantly from that in~\cite{Bacci2023}. Rather than using linear algebra, we use Krivine-Stengle's Positivstellensatz~\cite{Stengle74} (a variant of Hilbert's Nullstellensatz), promoting  further the connection between arithmetic and logical reasoning. These  theorems are also used in algebraic proof complexity, see, \textit{e.g.}, \cite{Beame96,GrigorievV01,Pit16}.

\item Unlike $\al$, $\pl$ allows the ``Booleanization'' formulas and judgements, which we use to prove a deduction theorem (Theorem~\ref{th:deductionThm}) that was not available in $\al$. 

\item The completeness proof employs a polynomial-time non-de\-ter\-min\-is\-tic reduction (in the sense of Adleman and Manders~\cite{AdlemanM79}) that translates any $\pl$ inference to a set of inferences in a specific normal form. The reduction presented in this paper can be instantiated to fragments of $\pl$, such as $\al$. Notably, when instantiated to $\al$, it both generalises and significantly simplifies the normalisation algorithm proposed in~\cite{Bacci2023}.
We speculate that this technique can be a valid tool to obtain, and/or simplify, other completeness proofs. 

\item Relying on the reduction discussed above, we establish new complexity results for two fundamental decision problems: satisfiability of a set of judgements and semantical consequence from a set of judgements. This is done for both $\al$ and $\pl$. We prove that deciding satisfiability is NP-complete in $\al$ and in PSPACE for $\pl$ (Theorems~\ref{thm:SatALNP-compl}--\ref{SATpl-PSPACE}); and that deciding semantical consequence is co-NP complete in $\al$ and in PSPACE for $\pl$ (Theorems~\ref{thm:cosequenceAPisCoNP}--\ref{SCpl-PSPACE}).
\end{enumerate}

\textbf{Related Work.}
Parallel to Lawvere's real-valued approach, we must mention the works on graded linear logics~\cite{DagninoP22,LemayV23} and normed logics~\cite{Grandis2007}, and the vast development of fuzzy logics~\cite{HajekBook,P79}. Among the latter, product logic~\cite{H96,H06,Sc06,EG00}, defined over the multiplicative quantale $[0,1]$, is the closest to ours. Through the quantale isomorphism $e^{-x}$, $\pl$ can be viewed as an extension of product logic: tensor corresponds to product $a \cdot b$, and multiplication to the operation $a \odot b \coloneq a^{-\ln b}$. 
However, such an interpretation of the logical connectives seems unnatural for quantitative reasoning, and impedes us from direct access to results we make use of, \textit{e.g.}, in algebraic geometry (such as the Krivine-Stengle Positivstellensatz, used for completeness)  and in linear algebra (such as  Khachiyan's ellipsoid method, used for complexity).

\textbf{Synopsis.} Section~\ref{sec:prelim} gives  preliminary definitions and notation. 
Section~\ref{sec:pl} gives the syntax and semantics of $\pl$,
and Section~\ref{sec:deductionsystem} presents a deduction system for it. 
Section~\ref{sec:applications} presents some nontrivial applications. 
Section~\ref{sec:canonicalforms} discusses canonical forms for $\pl$'s judgements, and 
Section~\ref{sec:completeness} develops the completeness result. 
Section~\ref{sec:complexity} gives the complexity results for $\pl$ and its affine fragment $\al$. 
Finally, Section~\ref{sec:concl} gives concluding remarks and discusses future work. (Appendix~\ref{appendix} gives some
detailed proofs.)

\section{Preliminaries and Notation} \label{sec:prelim}

A \emph{quantale} is a complete lattice with a binary, associative operation $\ot$ (\emph{tensor}), such that for every element $a$, both $a \ot -$ and $- \ot a$ have right adjoints (equivalently, $\ot$ preserves all joins). 
A quantale is \emph{commutative} whenever its tensor is; and \emph{unital} if there
is an element $u$ (\textit{unit}) s.t. $u \ot a = a = a \ot u$, for all $a$; when the unit is the top element, the quantale is \emph{integral}.
For commutative quantales we denote the right adjoint to $a \ot -$ by $a \lol -$. Hence, $a \ot b \leq c  \Longleftrightarrow  b \leq a \lol c$.

Commutative integral quantales are 
(i) the \textit{Boolean quantale} 
with logical conjunction as tensor; 
(ii) the \textit{{\L}ukasiewicz quantale} on $[0,1]$ 
with truncated addition as tensor; and (iii) the \emph{Lawvere quantale} over the complete lattice of the extended positive reals $\extreals$ taken with reverse order 
and with sum as tensor. The first two have been used 
for interpreting Boolean and {\L}ukasiewicz  logics~\cite{Lukasiewicz30}; the third one to interpret Lawvere logics~\cite{Bacci2023}.

In this paper, we work with sum $+$ and multiplication $\cdot$ on $\extreals$ as defined in Table~\ref{tab:operation}. They form a semiring, i.e., two commutative monoids with $0$ and $1$ as their identity elements, with multiplication distributing over the sum, and $0 \cdot x = 0$. The sum is the sum of the Lawvere quantale. With the inverse order on the positive reals, multiplication does not preserve all joins, although it does w.r.t.\ to the usual order on the reals. We could have made other choices. If we instead put \mbox{$0 \cdot \infty = \infty$} we would obtain a multiplication that preserves all joins w.r.t.\  the inverse order, but we would not obtain a semiring. If we worked with the usual order on the reals we would obtain a semiring in the category of complete lattices, but we would no longer be working with Lawvere's quantale.   So, while no choice is perfect, ours maintains a connection with Lawvere's quantale and we do have a semiring. Having a semiring is particularly useful as it enables us to encode examples from measure theory (Section~\ref{sec:applications}) and to obtain a deduction theorem (Theorem~\ref{th:deductionThm}) (which we would otherwise lack).

\begin{table*}[tb!]
\begin{align*}
    \begin{tabular}{c|ccccc}
		$+$ & $0$ & $s$ & $\infty$ \\ 
		\hline 
		$0$ &  $0$ & $s$ & $\infty$ \\ 
		$r$ &  $r$ & $~~r+s~~$ & $\infty$ \\ 
		$\infty$ & $\infty$ & $\infty$ & $\infty$
	\end{tabular}
    &&
	\begin{tabular}{c|ccccc}
		$\cdot$ & $0$ & $s$ & $\infty$ \\ 
		\hline 
		$0$ &  $0$ & $0$ & $0$ \\ 
		$r$ &  $0$ & $~~r \cdot s~~$ & $\infty$ \\ 
		$\infty$ & $0$ & $\infty$ & $\infty$
	\end{tabular}
    &&
    \begin{tabular}{c|ccccc}
		$\dotdiv$ & $0$ & $s$ & $\infty$ \\ 
		\hline 
		$0$ &  $0$ & $0$ & $0$ \\ 
		$r$ &  $r$ & $~~r - s~~$ & $0$ \\ 
		$\infty$ & $\infty$ & $\infty$ & $0$
	\end{tabular} 
\end{align*}

\caption{Tabular definition for sum, multiplication, and truncated subtraction on $\extreals$ (first argument in the leftmost column; second in the top row; $r,s\in(0,\infty)$).}
\label{tab:operation}
\end{table*} 

We use $\leq$, $\geq$, $\inf$, $\sup$ with the standard ordering of the extended reals to avoid confusion. However, being the reversed order of the Lawvere quantale, join and meet correspond to $\inf$ and $\sup$, respectively, with $\infty$ as the bottom element and $0$ as the top. The right adjoint $s \lol r$ is just the inverted truncated subtraction $r \dotdiv s$ defined in Table~\ref{tab:operation}.


\section{Polynomial Lawvere Logic}
\label{sec:pl}

In this section, we introduce \emph{Polynomial Lawvere logic} ($\pl$), a propositional logic interpreted over the  semiring. It extends Affine Lawvere logic ($\al$)~\cite{Bacci2023} by allowing the multiplication of formulas, thereby enabling the encoding of polynomial functions over the positive extended reals. As a result, $\pl$ offers a more expressive logic, capable of handling more complex quantitative reasoning.

\textbf{Syntax.} Let $\Prop$ be a set of propositional variables. 
The formulas of $\pl$ are freely generated by the following grammar.
\begin{align*}
    \phi,\psi ::= \bot \mid x \mid r \mid
    \phi \ot \psi \mid \phi \lol \psi \mid \phi \psi
    && \text{(for $x\in\X$ and  $r \in \preals$)}
\end{align*}

For $n \in \naturals$, we define $\phi^n$ inductively as
$\phi^0 \coloneq 1$ and $\phi^{(n+1)} \coloneq \phi \phi^n$.
Moreover, from the quantale connectives, one can derive the usual logical connectives
\begin{align*}
\begin{aligned}
	&\top \coloneq \bot \lol \bot \,,\\
	&\lnot\phi \coloneq \phi\lol\bot \,,
\end{aligned}
&&
\begin{aligned}
    &\phi \land \psi \coloneq \phi \ot (\phi \lol \psi) \,, \\
	&\phi \lor \psi \coloneq ((\psi \lol \phi) \lol \phi) \land ((\phi \lol \psi) \lol \psi) \,, \\
	&\phi \lollol \psi \coloneq (\phi \lol \psi) \land (\psi \lol \phi) \,.
\end{aligned}
\end{align*}

Let $\phi[\psi/x]$ be the result of substituting the variable $x\in \X$ for the formula $\psi$ in $\phi$. This notation extends canonically to lists and sets of formulas. 

We assume all binary operators are left-associative and follow an \textbf{operator precedence rule}: multiplication has the highest precedence, followed by $\lnot$, then $\ot$, next $\land$ and $\lor$, and finally $\lol$ and $\lollol$ with the lowest precedence.
 Thus, the formula $\theta\phi\ot\psi\land \lnot\theta\psi\lol\theta$ is interpreted as $(((\theta\phi)\ot\psi)\land (\lnot(\theta\psi)))\lol\theta$

\textbf{Semantics.} The models are maps $\M \colon \Prop \to \extreals$ interpreting the 
propositional variables in our semiring, extended to all formulas as follows 
\begin{align*}
	\begin{aligned}
		\M(\bot) &\coloneq \infty \,, \\
		\M(r) &\coloneq r \,, \\
		\M(\phi\psi) &\coloneq \M(\phi)\M(\psi)\,,
	\end{aligned}
	&&
	\begin{aligned}
		\M(\phi \ot \psi) &\coloneq \M(\phi) + \M(\psi) \,, \\
		\M(\phi \lol \psi) &\coloneq \M(\psi) \dotdiv \M(\phi) \,, 
	\end{aligned}
\end{align*}
Consequently, the derived connectives are interpreted as follows.
\begin{align*}
\begin{aligned}
  \M(\top) 			&= 0 \,, \\
  \M(\lnot\phi) 		&= \infty - \M(\phi) \,, 
\end{aligned}
  &&
\begin{aligned}
  \M(\phi \land \psi)	&= \max \{ \M(\psi), \M(\phi) \} \,, \\
  \M(\phi \lor \psi)	&= \min \{ \M(\psi), \M(\phi) \}  \,, \\
  \M(\phi \lollol \psi) &= |\M(\phi) - \M(\psi)| \,.
\end{aligned}
\end{align*}

\textbf{Affine Lawvere Logic} ($\al$), introduced in \cite{Bacci2023}, follows the grammar below and shares the same semantics as $\pl$, appropriately restricted.
\begin{align*}
	\phi,\psi ::= \bot \mid x \mid r \mid
	\phi \ot \psi \mid \phi \lol \psi \mid r \psi
	&& \text{(for $x\in\X$ and  $r \in \preals$)}
\end{align*}

\textbf{Boolean formulas.} Observe that while, in general, a model evaluates a formula to any value in $[0,\infty]$, formulas such as $\lnot\phi$ or $\phi\bot$ can either be evaluated to $0$ (``true'') or to $\infty$ (``false''). For instance,
\begin{align*}
    \begin{aligned}
    \M(\lnot\phi) = 
		\begin{cases}
			0		&\text{if $\M(\phi)$ is infinite} \\
			\infty	&\text{otherwise} \,,
		\end{cases}  
    \end{aligned}
    &&
    \begin{aligned}
    	\M(\lnot\lnot\phi) = 
    	\begin{cases}
    		0		&\text{if $\M(\phi)$ is finite} \\
    		\infty	&\text{otherwise} \,.
    	\end{cases}  
    \end{aligned}
\end{align*}
We call such formulas \emph{Boolean}, and they define more useful derived operators. 
\begin{align*}
\begin{aligned}
	&\phi=\psi \coloneq (\phi\lollol\psi)\bot \,, \\
	&\phi \neq \psi \coloneq \lnot(\phi\lollol\psi)\bot \,,
\end{aligned}
&&
\begin{aligned}
    &\phi\geq\psi \coloneq (\phi\lol\psi)\bot \,, \\
	&\phi>\psi \coloneq \lnot(\psi\lol\phi)\bot \,,
\end{aligned}
&&
\begin{aligned}
    &|\phi| \coloneq \bot > \phi 
    \,.
\end{aligned}
\end{align*}
These have the expected ``Boolean'' meaning:
\begin{align*}
    \begin{aligned}
    \M(\phi = \psi) = 
		\begin{cases}
			0		&\text{if $\M(\phi)=\M(\psi)$} \\
			\infty	&\text{otherwise} \,,
		\end{cases}
    \\
    \M(\phi \neq \psi) = 
		\begin{cases}
			0		&\text{if $\M(\phi)\neq\M(\psi)$} \\
			\infty	&\text{otherwise} \,,
		\end{cases}
    \end{aligned}
    &&
    \begin{aligned}
    \M(\phi > \psi) = 
		\begin{cases}
			0		&\text{if $\M(\phi) > \M(\psi)$} \\
			\infty 	&\text{otherwise} \,,
		\end{cases}
    \\
    \M(\phi \geq \psi) = 
		\begin{cases}
			0		&\text{if $\M(\phi) \geq \M(\psi)$} \\
			\infty 	&\text{otherwise} \,.
		\end{cases}
    \end{aligned}
\end{align*}
Using them, we can express useful facts about our models, \eg, $|\phi|$ says that ``\textit{$\phi$ is finite}'' and $\phi > 0$ that ``\textit{$\phi$ is strictly positive}''. 

Hereafter, we use $\phi\leq\psi$ and $\phi<\psi$ as synonyms for $\psi\geq\phi$ and $\psi>\phi$ respectively. Also, concatenation of these formulas are used to express their conjunction: for example, $\phi\leq\psi<\theta$ means $(\phi\leq\psi)\land(\psi<\rho)$.

\textbf{Judgements.} A \emph{judgement} in $\pl$ is a syntactic construct of the form
\begin{equation*}
	\phi_1, \dots, \phi_n \vdash \psi \,,
	\tag{Judgement}
\end{equation*}
where $\phi_i$, $\psi$ are logical formulas. The antecedent $\phi_1, \dots, \phi_n$ of a judgement is a finite ordered list, possibly, with repetitions. 
As customary, for $\Gamma$ and $\Delta$ lists of formulas, their comma-separated juxtaposition $\Gamma, \Delta$ denotes concatenation; and $\vdash \phi$ is the notation for a judgement with empty list of antecedents.

A judgement $\Gamma \vdash \psi$ \emph{is satisfied by} a model $\M$, denoted $\Gamma \models_\M \psi$, whenever
\begin{equation*}
	\textstyle
	\sum_{\phi \in \Gamma} \M(\phi) \geq \M(\psi) \,.
	\tag{Semantics of judgements}
\end{equation*}
A judgement is \emph{satisfiable} if it is satisfied by a model; \emph{unsatisfiable} if it is not satisfiable; and a \emph{tautology} if it is satisfied by all models.

In particular, $\vdash \phi \lol \phi$, $\vdash \top$, and 
$\vdash \lnot\lnot \phi \lollol (\bot > \phi)$ are examples of tautologies, while $\vdash \phi \lollol (\lnot\lnot\phi)$
is not.

Note the distinction between $\phi\lol\psi$ and the Boolean formula $\phi\geq\psi$: while for all models $\M$, we have $\models_\M\phi\lol\psi$ iff $\models_\M\phi\geq\psi$, it may not hold that  $\M(\phi\lol\psi)=\M(\phi\geq\psi)$, as $\M(\phi\lol\psi)$ could be a non-zero finite number.

\begin{definition}[Semantic Consequence]
A judgement $\gamma$ is a \emph{semantic consequence} of a set $S$ of judgements, in symbols $S \models \gamma$, 
if every model that satisfies all the judgements in $S$ also satisfies $\gamma$.
\end{definition}


\section{Deduction System for $\pl$}
\label{sec:deductionsystem}

An \textit{inference (rule)} is a syntactic construct of the form 
\begin{equation*}
\infrule{\; S \;}{\gamma}
\end{equation*}
for $S$ a set of judgements and $\gamma$ a judgement. 
The judgements in $S$ are the \emph{hypotheses of the inference} and 
$\gamma$ is the \emph{conclusion}. When $S=\{\gamma'\}$ is a singleton, we write 
\begin{equation*}
\begin{aligned}
\doubleinfrule{\gamma'}{~\gamma~}
\end{aligned}
\quad \text{to denote both} \quad
\begin{aligned}
\infrule{\gamma'}{~\gamma~}
\end{aligned}
\;
\text{and }
\;
\begin{aligned}
\infrule{\gamma}{~\gamma'~}
\end{aligned}
\end{equation*}
and say that $\gamma$ is \emph{provably equivalent} to $\gamma'$.

\begin{table*}[tb!]	
	\begin{align*}
		\begin{array}{c}
			\begin{aligned}
				\infrule[id]{}{\phi \vdash \phi}
				&&&
				\infrule[top]{}{ \phi\vdash \top}
				&&&
				\infrule[bot]{}{\bot \vdash \phi}
			\end{aligned}
			\\[2ex]
			\begin{aligned}
				\infrule[cut]{
					\Gamma \vdash \phi
					&
					\Delta,\phi \vdash \psi
				}{ \Gamma,\Delta \vdash \psi }
				&&&
				\infrule[perm]{
					\Gamma, \phi, \psi, \Delta \vdash \theta
				}{\Gamma, \psi, \phi, \Delta \vdash \theta}
				&&&
				\infrule[weak]{
					\Gamma \vdash \phi
				}{\Gamma,\psi \vdash \phi}
			\end{aligned}
			\\[2ex]
			\textsc{1. Structural rules}			
			\\[2ex]
			\begin{aligned}
				\infrule[0-1]{}{ \vdash 0\land (0<1<\bot)}
				&&&
				\infrule[wem]{}{\vdash (\lnot\phi) \lor (\lnot\lnot\phi)}
				&&&
				\infrule[tot]{}{\vdash (\phi \lol \psi) \lor (\psi \lol \phi)}
			\end{aligned}
			\\[2ex]
			\begin{aligned}
				\doubleinfrule[$\ot_1$]{
					\hfill\Gamma, \phi, \psi \vdash \theta
				}{\Gamma, \phi \ot \psi \vdash \theta}
				&&&
				\infrule[$\lol_1$]{
					\Gamma, \phi \lol \theta \vdash \psi
					&
					\vdash\theta \geq \phi
				}{\Gamma, \theta \vdash \phi \ot \psi}
			\end{aligned}
			\\[2ex]
			\begin{aligned}
				\doubleinfrule[$\ot_2$]{
					\Gamma, \phi \ot \psi \vdash \theta
				}{\Gamma, \phi \vdash \psi \lol \theta}
				&&&
				\infrule[$\lol_2$]{
					\Gamma, \theta \vdash \phi \ot \psi
					&
					\vdash |\phi\lor \theta|
				}{\Gamma, \phi \lol \theta \vdash \psi}
			\end{aligned}
			\\[2ex]
			\textsc{2. Lawvere quantale rules}
			\\[2ex]
			\begin{aligned}
				\infrule[unit]{
				}{\vdash 1\phi\lollol\phi}
				&&&
				\infrule[zero]{
				}{\vdash 0\phi \lollol 0}
				&&& 
				\infrule[nullify]{\vdash\phi^{n+1}}{\vdash\phi}
			\end{aligned}
			\\[2ex]
			\begin{aligned}
				\infrule[comp]{\phi\vdash\psi}{\theta\phi\vdash\theta\psi}
				&&&
				\infrule[decomp]{\theta\phi\vdash\theta\psi & \vdash \bot > \theta  > 0}{\phi\vdash\psi}
			\end{aligned}
			\\[2ex]
			\begin{aligned}
				\infrule[a1]{
				}{\vdash(\phi\psi)\theta\lollol\phi(\psi\theta)}
				&&&
				\infrule[a2]{
				}{\vdash(r \cdot s)\phi\lollol r(s\phi)}
				&&&
				\infrule[comm]{
				}{\vdash \phi\psi\lollol\psi\phi}
			\end{aligned}
			\\[2ex]
			\textsc{3. Multiplicative rules}
			\\[2ex]
			\begin{aligned}
				\infrule[d1]{}{\vdash\theta(\phi\ot\psi)\lollol\theta\phi\ot\theta\psi}
				&&&
				\infrule[d2]{\psi>\phi>0\vdash |\theta|}{\vdash\theta(\phi\lol\psi)\lollol(\theta\phi\lol\theta\psi)} 
			\end{aligned}
			\\[2ex]
			\begin{aligned}
				\infrule[d3]{}{\vdash (r+s)\phi\lollol r\phi\ot s\phi}
				&&&
				\infrule[d4]{}{\vdash (r\dotdiv s)\phi\lollol (s\phi\lol r\phi)} 
			\end{aligned}
			\\[2ex]
			\textsc{4. Distributive rules}
		\end{array}
	\end{align*}
	\caption{Deduction system for polynomial Lawvere logic $\pl$.
		In the above, $\phi,\psi,\theta$ are formulas, $\Gamma, \Delta$ list of formulas, $r,s\in[0,\infty)$ positive reals, and $n$ a positive integer.}
	\label{tab:rules}
\end{table*}

The deduction system for $\pl$ is given in Table~\ref{tab:rules}.
It contains the basic inference rules of logical deduction (\textsc{id}) and (\textsc{cut}), and the structural rules of weakening (\textsc{weak}) and permutation (\textsc{perm}) (note that cancellation is not sound). It has rules reflecting the fact that the Lawvere quantale is integral, but also rules specific to Lawvere quantale alone: (\textsc{0-1}) states that $0$ is the bottom element and $1$ is neither the bottom nor the top element, (\textsc{wem}) is the weak excluded middle; (\textsc{tot}) states that the quantale is totally ordered; ($\ot_1$) states that $\ot$ behaves as additive conjunction; ($\ot_2$) is the adjunction rule for $\ot$ and $\lol$; while ($\lol_1$), ($\lol_2$) complement the adjunction rule by expressing the interactions in Lawvere quantale between the connectives $\ot$ and $\lol$ on opposite sides of $\vdash$. Note that ($\lol_2$) is conditional on the finiteness of specific formulas.
Lastly, there are the semiring rules together with rules for multiplication and distributivity. Some rules are specific to our semiring: for instance, (\textsc{d2}) is conditional to the order between $\phi$ and $\psi$ and the finiteness of $\theta$. 


\begin{definition}[Provability]
Let $S$ be a set of judgements. 
We say that a judgement $\gamma$ is \emph{provable from} 
(or \emph{deducible from}) $S$, if there exists a 
sequence $\gamma_1, \dots, \gamma_n$ of judgements 
ending in $\gamma$ whose members are either 
members of $S$, or each follows from some preceding members of the sequence by using the inference rules of the deduction system.
A sequence $\gamma_1, \dots, \gamma_n$ as above is called \emph{proof}. 
\end{definition}

In what follows we will (safely) abuse notation: if $\gamma$ is a judgement, $S$ a set of judgements and $\M$ a model, we write
$\frac{\; S \;}{\gamma}$, if $\gamma$ is provable from $S$;
$\M \models S$ if $\M$ is a model of all judgements in $S$. If $S=\{\gamma\}$, we write $\M\models\gamma$.

\begin{theorem}[Soundness]\label{soundness}
If a judgement $\gamma$ is provable from $S$ in $\pl$, then $\gamma$ is a semantic consequence of $S$. In symbols:
\begin{align*}
  \frac{\; S \;}{\gamma} 
  &&\text{implies} &&
  S \models \gamma  \,.
\end{align*}
\end{theorem}


Any judgement $\phi_1, \dots, \phi_n \vdash \psi$ is provably equivalent to $\phi_1 \ot \dots \ot \phi_n \vdash \psi$. Moreover, any judgement of type $\phi \vdash \psi$ is provably equivalent to $\vdash \phi \lol \psi$. Hence, without loss of generality, we may assume that arbitrary judgements in $\pl$ are of the form $\vdash\theta$, for some formula $\theta$.

In~\cite{Bacci2023} it is shown that $\al$ does not enjoy a deduction theorem, not even in the weak form that holds for fuzzy logics, such as {\L}ukasiewicz, G{\"o}del, or product logics. This is because we have proven that in $\al$ it is not possible to ``internalize'' provability in the language of the logic. However, in $\pl$, the expressivity provided by multiplication allows us to do it by ``Booleanizing'' the judgements.
\begin{theorem}[Deduction Theorem] \label{th:deductionThm}
For any formulas $\phi,\psi$ in $\pl$, 
\begin{align*}
	\frac{\; \vdash\phi \;}{\vdash\psi} 
	&& \text{iff} && \frac{\; \vdash 0 \geq \phi \;}{\vdash0 \geq \psi}  && \text{iff} && \vdash(0 \geq \phi)\lol(0 \geq \psi)\,.
\end{align*}
\end{theorem}

 

\section{Applications: Proving Properties of Distances} 
\label{sec:applications}

In this section, we show how the deductive system for $\pl$ 
can be used to reason about the properties of two well-known distances on probability distributions, namely, the total variation and the Kantorovich distance, and we discuss embedding quantitative equational logic in $\pl$.

Let $X = \{x_1, \dots, x_n\}$ be a finite (extended) metric space with distances $d_{ij}$ between $x_i$ and $x_j$ possibly taking $\infty$ as value.
Denote by $\mu$, $\nu$, $\rho$ be three generic discrete probabilities on $X$ and by $\mu_i$, $\nu_i$, $\rho_i$ their probabilities at $x_i \in X$.

\smallskip
\textbf{Total Variation.} The total variation distance between $\mu$ and $\nu$ is defined as 
$d_{TV}(\mu,\nu) = \max_{A\subseteq X}|\mu(A) - \nu(A)|$,  which can be encoded in $\pl$ by the formula 
$t_{\mu,\nu} \coloneq \bigwedge_{A \subseteq \{1..n\}} (\bigoplus_{i\in A} \mu_i \lollol \bigoplus_{i\in A} \nu_i)$. 
A simple example to start with is to demonstrate that the total variation is a pseudo-metric, \ie, satisfies the axioms of reflexivity, symmetry, and triangle inequality, which can be expressed in $\pl$: 
\begin{align*}
    (\textsc{refl}) \; \vdash t_{\mu,\mu} 
    &&
    (\textsc{symm}) \; t_{\mu,\nu} \vdash t_{\nu,\mu} \,,
    &&
    (\textsc{triang}) \; t_{\mu,\nu} , t_{\nu,\rho} \vdash t_{\mu,\rho} \,.
\end{align*}
The first two are trivial to derive. The derivation of the third is shown below:
\begin{equation*}
\adjustbox{max width=0.95\textwidth}{%
\small
\infrule[def, $\ot_1$]{
\infrule[$\ot_1,\land_1$,$\land_2$]{
\infrule[$\land_2$]{
\infrule[$\ot_1,\land_1$]{
    \infrule[$\ot_2$]{
        \infrule[$\ot_1,\lol_2$]{
            \infrule[id]{}{
            \mu_i\ot\nu_i\ot\rho_i \vdash \mu_i\ot\nu_i\ot\rho_i}
        }{\mu_i \ot (\mu_i\lol\nu_i) \ot (\nu_i\lol\rho_i) \vdash \rho_i}
    }{(\mu_i\lol\nu_i) \ot (\nu_i\lol\rho_i) \vdash \mu_i\lol\rho_i }
}{(\mu_i\lollol\nu_i) \ot (\nu_i\lollol\rho_i) \vdash \mu_i\lol\rho_i}
&
\infrule[$\ot_1,\land_1$]{
    \infrule[$\ot_2$]{
        \infrule[$\ot_1,\lol_2$]{
            \infrule[id]{}{
            \rho_i\ot\nu_i\ot\rho_i \vdash \mu_i\ot\nu_i\ot\rho_i}
        }{\rho_i \ot (\nu_i\lol\mu_i) \ot (\rho_i\lol\nu_i) \vdash \mu_i}
    }{(\nu_i\lol\mu_i) \ot (\rho_i\lol\nu_i) \vdash \rho_i\lol\mu_i }
}{(\mu_i\lollol\nu_i) \ot (\nu_i\lollol\rho_i) \vdash \rho_i\lol\mu_i}
}{(\mu_i\lollol\nu_i) \ot (\nu_i\lollol\rho_i) \vdash \mu_i\lollol\rho_i}
}{%
\bigwedge_{A \subseteq \{1..n\}} (\bigoplus_{i\in A} \mu_i \lollol \bigoplus_{i\in A} \nu_i)
\ot 
\bigwedge_{A \subseteq \{1..n\}} (\bigoplus_{i\in A} \nu_i \lollol \bigoplus_{i\in A} \rho_i)
\vdash 
\bigwedge_{A \subseteq \{1..n\}} (\bigoplus_{i\in A} \mu_i \lollol \bigoplus_{i\in A} \rho_i)}
}{t_{\mu,\nu} , t_{\nu,\rho} \vdash t_{\mu,\rho}}
}
\end{equation*}
Note that some steps of the derivation use meta-rules which are derivable from the rules in Table~\ref{tab:rules}, such as ($\land_1$), ($\land_2$) (see Appendix~\ref{app:useful}).

The total variation is not just a pseudo-metric, but a proper metric satisfying the
Fr\'echet positivity axiom, which can be expressed in $\pl$ by the judgement
\begin{equation*}
    (\textsc{positivity}) \; \bigwedge (\mu_i \neq \nu_i) \vdash (t_{\mu,\nu} > 0) \,.
\end{equation*}
Observe that the above uses the Boolean formulas of $\pl$, which can be expressed using multiplication by $\bot$. In fact, this is a non-linear property that cannot be captured by $\al$ as it allows only affine formulas.

\smallskip
\textbf{Kantorovich distance.} 
The Kantorovich distance%
\footnote{Also known as the Wasserstein distance or Earth's mover distance.}
between $\mu$ and $\nu$ can be defined using the following two equivalent (dual) formulations
\begin{equation*}
    d_K(\mu,\nu) = \inf_{\omega} \sum_{i,j} \omega_{ij} d_{ij} 
    = \sup_{f} \Big|\sum_i f_i \mu_i - \sum_i f_i \nu_i \Big|
    \tag{K-R duality}
\end{equation*}
where $\omega$ ranges over joint probability distributions with $\mu$ as left-marginal (\ie, $\sum_j \omega_{ij} = \mu_i$, for all $i$) and $\nu$ as right-marginal (\ie, $\sum_i \omega_{ij} = \nu_j$, for all $j$);
and $f$ over non-expanding $\preals$-valued maps on $X$, \ie, $|f_i - f_j| \leq d_{ij}$, for all $i,j$.

As its definitions involve $\inf$ (infimum) on one hand, and $\sup$ (supremum) on the other hand, we cannot express the Kantovich distance as a single formula in $\pl$. However, we should not despair as we can still reason about it if we can find a finite set of judgements that uniquely characterises its value. The set we propose, hereafter denoted by $\mathcal{K}$, contains the
following judgements:
\begin{gather*}
    \vdash 
    \bigwedge_i (\bigoplus_j W_{ij} \lollol \mu_i) 
    \land
    \bigwedge_j (\bigoplus_i W_{ij} \lollol \nu_j) \,,
    \quad
    \vdash \bigwedge_{i,j} \big( d_{ij} \lol (F_j \lol F_i) \big)
    \land \bigwedge_i |F_i| \,,
    \\
    \begin{aligned}
        \bigoplus_i F_i \mu_i \lollol \bigoplus_i F_i \nu_i \vdash K_{\mu,\nu} \,, 
        &&
        K_{\mu,\nu} \vdash \bigoplus_{i,j} W_{ij} d_{ij} \,,
    \end{aligned}
\end{gather*}
where $W_{ij}$, $F_i$, and $K_{\mu,\nu}$ are propositional variables. 
This set is derived by following the steps of the proof of (strong) duality in linear programs~\cite{Schrijver1998}, specifically tailored to the K-R duality presented above. The first two judgments represent the conjunction of the constraints from both the primal and dual linear programs (\ie, the marginal conditions on $\omega$ and the non-expanding condition on $f$). The last two imply $\bigoplus_i F_i \mu_i \lollol \bigoplus_i F_i \nu_i \vdash \bigoplus_{i,j} W_{ij} d_{ij}$, corresponding to the optimality condition for the feasible solutions. The variable $K_{\mu,\nu}$ is a  convenience.

%
This encoding is such that all the models that satisfy $\mathcal{K}$ assign the variable $K_{\mu,\nu}$ value $d_K(\mu,\nu)$, \textit{i.e.}, the Kantorovich distance between $\mu$ and $\nu$.
Indeed, next we show that from $\mathcal{K}$ we can deduce  
\begin{align}
    \vdash K_{\mu,\nu} \lollol \Big(\bigoplus_i F_i \mu_i \lollol \bigoplus_i F_i \nu_i  \Big)
    &&\text{and}&&
    \vdash K_{\mu,\nu} \lollol \bigoplus_{i,j} W_{ij} d_{ij}  \,.
    \label{eq:k}
\end{align}
The above follows by deriving the following two judgments from $\mathcal{K}$
\begin{align*}
    \bigoplus_{i,j} W_{ij} d_{ij}  \ot \bigoplus_i F_i \mu_i \vdash \bigoplus_j F_j \nu_j \,,
    &&
    \bigoplus_{i,j} W_{ij} d_{ij} \ot 
    \bigoplus_j F_j \nu_j \vdash \bigoplus_i F_i \mu_i
\end{align*}
as they imply $\bigoplus_{i,j} W_{ij} d_{ij} \vdash \bigoplus_i F_i \mu_i \lollol \bigoplus_i F_i \nu_i$. Note that this corresponds to the steps of the proof of weak duality in linear programs. We show only the derivation of the first one as the other is similar. Below we provide only the schematic steps of the derivation, which would otherwise take too much space
\begin{align*}
\bigoplus_{i,j} W_{ij} d_{ij} \ot \bigoplus_i F_i \mu_i 
&\vdash \bigoplus_{i,j} W_{ij} d_{ij} \ot \bigoplus_i F_i (\bigoplus_j W_{ij})) 
\tag{left-marginal} \\
&\vdash \bigoplus_{i,j} F_j W_{ij}
\tag{\textsc{d1},$\ot_1$,\textsc{perm},non-expanding}\\
&\vdash \bigoplus_j F_j \nu_j
\tag{\textsc{d1}, right-marginal}
\end{align*}
In the above a concatenation of the form $\phi \vdash \psi \vdash \vartheta$ means that 
both $\phi \vdash \psi$ and $\psi \vdash \vartheta$ are derivable; the desired result follows by repeated applications of (\textsc{cut}).

Now that we have established a way to encode the Kantorovich distance, we can prove some of its properties. A well-known result from~\cite{Gibbs02} relating the 
Kantorovich distance with the total variation is 
\begin{equation*}
    d_K(\mu,\nu) \geq d_{\min} \cdot d_{TV}(\mu,\nu) \,,
\end{equation*}
where $d_{\min} = \min_{i \neq j}d_{ij}$. According to our encoding, such a statement is equivalent to establishing the following inference
\begin{equation*}
    \infrule{\mathcal{K}}{ K_{\mu,\nu} \vdash (\bigvee_{i \neq j}d_{ij}) t_{\mu,\nu}}
\end{equation*}
Due to lack of space, below we provide only the sketch of the proof. The key steps of it are to show that the judgements below follow from $\mathcal{K}$ for all $A\subseteq \{1..n\}$
\begin{align*}
    \bigoplus_{i \neq j} W_{ij} \ot \bigoplus_{i \in A} \mu_i 
    \vdash
    \bigoplus_{i \in A} \nu_i 
    &&
    \bigoplus_{i \neq j} W_{ij} \ot \bigoplus_{i \in A} \nu_i 
    \vdash
    \bigoplus_{i \in A} \mu_i 
\end{align*}
from which, by using ($\ot_2$), ($\lor_2$), one gets $\bigoplus_{i \neq j} W_{ij} \vdash t_{\mu,\nu}$. Thus, by applying the deduction rules of $\pl$, \eqref{eq:k}, and the fact that $d_{ii} = 0$ for all $i$, we get 
\begin{equation*}
    K_{\mu,\nu} 
    \vdash \bigoplus_{i,j} W_{ij} d_{ij}
    \vdash \bigoplus_{i\neq j} W_{ij} d_{ij}
    \vdash \bigoplus_{i\neq j} W_{ij} (\bigvee_{i \neq j}d_{ij})
    \vdash (\bigvee_{i \neq j}d_{ij}) t_{\mu,\nu} \,.
\end{equation*}
The desired inference follows from the above by repeated applications of (\textsc{cut}).


\textbf{Quantitative Equational Logic (QEL).} 
Already in~\cite{Bacci2023} we have shown how one can embed QEL in $\al$: add variables $s=t$ for all the terms $s,t$ of a chosen quantitative algebra; encode quantitative equations $s=_\e t$ as judgements $\e\vdash s=t$ in $\al$; and represent the axioms of QEL as inference rules.  

Since $\al$ lacks a deduction theorem, the embedding of QEL in $\al$ relies on inference rules. However, in $\pl$, the inferences in~\cite[Table~2]{Bacci2023} can be formalized as proper judgments using the deduction theorem (Theorem~\ref{th:deductionThm}).
Additionally, while $\al$ can handle only affine functions, the axioms of barycentric algebras with the $p$-Wasserstein metric require polynomial functions~\cite{Mardare16}. These types of examples can now be expressed in $\pl$.

With this embedding we achieve an approximate completeness result: if, under the encoding, 
$s=_\e t$ is valid then $s=_{\e'} t$ is provable for all $\e' > \e$. An exact embedding would 
need one of QEL’s infinitary rule; we leave this for future work.



\section{Canonical \& Polynomial Forms}
\label{sec:canonicalforms}

We define canonical forms for $\pl$ formulas and their judgements and provide 
a method for reducing judgements to their canonical form. The canonicalisation of judgements is a crucial ingredient for our proof of completeness (Section~\ref{sec:completeness}) and also an essential tool for establishing the complexity results in 
Section~\ref{sec:complexity}. 

\begin{definition}[Canonical Forms]
A formula is in \emph{canonical form} (CF) if it is either $\bot$ or in \emph{proper canonical
form} (PCF), \ie, a formula generated by
\begin{align*}
   \theta &\coloneqq x \mid r \mid 
   \theta \ot \theta \mid \theta \lol \theta \mid \theta \theta \,
   && \text{for $x\in \X$ and $r \in \preals$} \,.
\end{align*}
A formula in PCF is in \emph{polymomial form} when it has no occurrences of $\lol$.
\end{definition}
Formulas in PCF are closed under conjunction, disjunction and double implication, but not negation. Observe that $\bot$ does not occur in formulas in PCF.

Canonical forms are extended to judgements in the obvious way: $\phi \vdash \psi$ is 
in a certain form (\eg, CF, PCF, or polynomial) if both $\phi$ and $\psi$ are in that form.

\subsubsection*{Provable Equalities between Polynomial Formulas.}
By identifying formulas up to commutativity-associativity of $\ot$ and multiplication (and alpha-conversion), we have that polynomial formulas on $n$-variables are in 1-1 correspondence with polynomials with positive coefficients in $n$-variables. For a polynomial formula $\theta$ with variables $x_1, \dots, x_n$, the corresponding polynomial $\sem{\theta} \colon \reals^n \to \reals$ is defined as follows, for $v = (v_1, \dots, v_n) \in \reals^n$
\begin{align*}
\begin{aligned}
  \sem{x_i}(v) &= v_i \,, \\
  \sem{r}(v) &= r \,,
\end{aligned}
&&
\begin{aligned}
  \sem{\theta \otimes \vartheta}(v) 
  	&= \sem{\theta}(v) + \sem{\vartheta}(v) \,, \\
  \sem{\theta \vartheta}(v) 
  	&= \sem{\theta}(v) \sem{\vartheta}(v) \,.
\end{aligned}
\end{align*}
It is not difficult to prove, inductively, that  
$\M(\theta) = \sem{\theta}(\M(x_1), \dots, \M(x_n))$, for all positive real-valued models $\M \colon \Prop \to \preals$.

\begin{lemma}\label{polid}
Let $\theta, \vartheta$ be in polynomial form. $\sem{\theta} = \sem{\vartheta}$ iff ${\vdash\theta\lollol\vartheta}$ is provable.
\end{lemma}

The following corollary allows us to reason about provable equalities between non-polynomial formulas in terms of equalities over
polynomial formulas.
\begin{corollary} \label{polidext}
Let $\theta$, $\vartheta$ be polynomial formulas such that $\sem{\theta} = \sem{\vartheta}$, 
$\phi$, $\psi$ two formulas and $S$ a finite set of judgements in $\pl$. Then,
\begin{align*}
\frac{S}{\phi\lollol\psi \vdash \theta\lollol\vartheta} \,\;\text{ and }\;\,
\frac{S}{\theta\lollol\vartheta \vdash \phi\lollol\psi} 
&&\text{implies} &&
\frac{S}{\vdash\phi \lollol \psi} \,.
\end{align*} 
\end{corollary}

\subsubsection*{Canonicalisation.}
For a formula $\phi$ in $\pl$, we define its canonical form $\phi^{\sf cf}$, by induction on the formula as follows:
\begin{gather*}
\begin{aligned}
	x^{\sf cf} = x
	&&
	\bot^{\sf cf} = \bot
	&&
	r^{\sf cf} = r 
\end{aligned}
\\
\begin{aligned}
   (\phi \ot \psi)^{\sf cf} &= 
  	\begin{cases}
	  \bot  &\text{if $\phi^{\sf cf} = \bot$ or $\psi^{\sf cf} = \bot$} \\
	  \phi^{\sf cf} \ot \psi^{\sf cf} &\text{otherwise}
	\end{cases}
   \\
   (\phi \lol \psi)^{\sf cf} &= 
  	\begin{cases}
	  0 &\text{if $\phi^{\sf cf} = \bot$} \\
	  \bot  &\text{if $\phi^{\sf cf} \neq \bot$ and $\psi^{\sf cf} = \bot$} \\
	  \phi^{\sf cf} \lol \psi^{\sf cf} &\text{otherwise}
	\end{cases}
   \\
   (\phi \psi)^{\sf cf} &= 
  	\begin{cases}
	  0	&\text{if $\phi^{\sf cf} = 0$ or $\psi^{\sf cf} = 0$} \\
	  \bot  &\text{if $\phi^{\sf cf},\psi^{\sf cf} \neq 0$ and ($\phi^{\sf cf} = \bot$ or $\psi^{\sf cf} = \bot$)} \\
	  \phi^{\sf cf} \psi^{\sf cf} &\text{otherwise}
	\end{cases}
\end{aligned}
\end{gather*}

Formulas in $\pl$ are not necessarily semantically equivalent to their canonical form, but they get the same value in all
$(0,\infty)$-valued models.
\begin{proposition} \label{prop:cfSemantic}
  For any formula $\phi$, 
   $\phi^{\sf cf}$ is in canonical form, and for all models $\M \colon \Prop \to (0,\infty)$, 
    $\M(\phi) = \M(\phi^{\sf cf})$.
\end{proposition}

Observe that the restriction to $\infty$-free models is necessary, as if  
$\M(x) = \infty$, $\M(x \lol \bot) = 0$ while $\M((x \lol \bot)^{\sf cf}) = \M(\bot) = \infty$. Similarly, $0$-free models are necessary, as for $\M(x) = 0$, $\M(x\bot) = 0$ and $\M((x\bot)^{\sf cf})= \M(\bot) = \infty$.

\begin{definition}[Canonicalisation]
The \emph{canonical form} of a judgement $\phi \vdash \psi$ (also referred to as \emph{canonicalisation}) is $\phi^{\sf cf} \vdash \psi^{\sf cf}$.
Canonicalisation is extended to sets of judgements by $S^{\sf cf} = \{ \phi^{\sf cf} \vdash \psi^{\sf cf} \mid \phi \vdash \psi \in S\}$.
\end{definition}


\begin{proposition}\label{prop:cfOFcf}
Both $\vdash \phi^{\sf cf} \lollol (\phi^{\sf cf})^{\sf cf}$ and $\phi^{\sf cf} \vdash \phi$
are valid and provable. Moreover, if $\phi$ is in CF then $\phi = \phi^{\sf cf}$.
\end{proposition}

\begin{remark} \label{rmk:equivCounterExpl}
The inferences below are both unsound, thus not provable in the deductive system for $\pl$.
\begin{align*}
\infrule[a]{\phi \vdash \psi}{\phi^{\sf cf} \vdash \psi^{\sf cf}} 
&&
\infrule[b]{\phi^{\sf cf} \vdash \psi^{\sf cf}}{\phi \vdash \psi}
\end{align*}
A counterexample for the soundness of (A) is when $\phi = x$ and $\psi = x\bot$: 
take $\M(x) = 0$, then $\M(x) \geq \M(x\bot) = 0$ but 
$\M(x^{\sf cf}) < \infty = \M(\bot) = \M((x\bot)^{\sf cf})$.
\\
A counterexample for the soundness of (B) is when $\phi = x \lol \bot$ and $\psi = x$: 
take $\M(x) = \infty$, then $\M((x \lol \bot)^{\sf cf}) = 0 \geq \M(x^{\sf cf})$,
but $\M(x \lol \bot) = 0 < \M(x)$.
\end{remark}

Intuitively, the canonicalisation of a judgement is better interpreted as an ``over approximation'', where the propositional variables are assumed to be always finite real-valued and strictly positive. Next,
we introduce a definition that will help us express
the concept sketched above. 

Fix the sets of judgements $\FF = \{ \vdash |x| \mid x \in \Prop \}$ 
and $\PP = \{ \vdash x > 0 \mid x \in \Prop \}$.
\begin{definition}[Restricted Provability]
%
%
%
For a judgment $\gamma$ and a set of judgments $S$, we say that $\gamma$ is \emph{$\preals$-provable} from $S$ if it is provable from $S \cup \FF$, \emph{$(0,\infty]$-provable} if it is provable from $S \cup \PP$, and \emph{$(0,\infty)$-provable} if it is provable from $S \cup \FF \cup \PP$.
\end{definition}
Intuitively, adding $\FF$ to the hypothesis asserts the finiteness of the variables, while $\PP$ ensures they are strictly positive.

\begin{proposition} \label{prop:fpProvabilityCF}
All judgements in $\pl$ are $(0,\infty)$-provably equivalent to their canonical forms, \ie, for all formulas $\phi$,$\psi$ in $\pl$
\begin{align*}
\frac{\phi \vdash \psi \quad \FF\cup\PP}{\phi^{\sf cf} \vdash \psi^{\sf cf}}
&& \text{and} &&
\frac{\phi^{\sf cf} \vdash \psi^{\sf cf} \quad \FF \cup \PP}{\phi \vdash \psi} \,.
\end{align*}
\end{proposition}



\section{Completeness and Incompleteness Results for $\pl$}
\label{sec:completeness}

We prove firstly that, in general, $\pl$ is incomplete.


\begin{theorem}[Incompleteness] \label{th:incompleteness}
The logic $\pl$ is incomplete, meaning that there exist a set of judgements $S$ and a judgement $\gamma$ such that $\gamma$ is a semantically consequence of $S$, but $\gamma$ is not provable from $S$ in $\pl$.
\end{theorem}
\begin{proof}
Let $x,y\in \X$ and $S=\{ (n+1)x \vdash ny\mid ~n \in \naturals\}$. Note that $x \vdash y$ is a semantical consequence of $S$. All the proofs in $\pl$ are finite, hence if there exists a proof for $x\vdash y$ from $S$, there must exists $k \in \naturals$ such that the only judgements used
in the proof are $V = \{ (n+1)x \vdash y \mid 0 \leq n \leq k \}$. If this is the case, any model of $V$ is a model of $x \vdash y$ (from soundness). But this is false: consider $\M$ such that
$\M(x) = \frac{k}{k+1}$ and $\M(y) = 1$; then, for all $n \leq k$, $\M((n+1)x) \leq \M(y)$, but 
$\M(x) < \M(y)$. Hence, $x \vdash y$ is not provable from $S$.
\qed\end{proof}

However, $\pl$ is complete if we only consider finitely axiomatized theories.


\begin{theorem}[Finite Completeness]
\label{thm:completeness}
Let $S$ be a finite set of judgements in $\pl$.
If a judgement $\gamma$ is a semantic consequence of $S$, then $\gamma$ is provable from $S$:
\begin{align*}
  S \models \gamma 
  &&\text{implies} &&
  \frac{\; S \;}{\gamma} \,.
\end{align*}
\end{theorem}

The proof plan is to reduce the statement above to the 
following restricted completeness theorem which applies only to judgements in polynomial form.
\begin{theorem}[Polynomial Completeness] \label{thm:completenessPoly}
Let $\gamma$ be a judgement and $S$ a finite set of judgements, both in polynomial form. Then,
\begin{align*}
  S \cup \FF \models \gamma 
  &&\text{implies} &&
  \frac{S \quad \FF}{\gamma} \,.
\end{align*}
\end{theorem}
Note that, $S \cup \FF \models \gamma$ represents a restricted form of semantical consequence where the models are assumed to be $\preals$-valued, which hereafter we denote by $S \models_{\preals} \gamma$.
We extend this notation to any $R \subseteq \extreals$, where $S \models_R \gamma$ denotes that every model $\M \colon \Prop \to R$ that satisfies $S$ satisfies also $\gamma$, for which we say that 
$\gamma$ is a \emph{$R$-semantical consequence} of $S$.

\smallskip
Before delving into the proof of Theorem~\ref{thm:completenessPoly} ---which constitutes the core of the completeness result--- we describe the reduction to it.

\medskip
The proposed reduction is characterised by a set of nondeterministic moves of the form
\begin{equation*}
	(S, \gamma) \longrightarrow (S_i, \gamma_i) \qquad \text{(for $i=1, \dots, k$)}
\end{equation*}
where $\gamma$, $\gamma_i$ are judgements and $S$, $S_i$ sets of judgements,
which have the following properties, called \emph{reliable} and \emph{nice}, respectively:
\begin{description}[fullwidth]
	\item[Reliable:] If $\gamma$ is a semantical consequence of $S$, then for all $i$, 
	$\gamma_i$ is a semantical consequence of $S_i$. In symbols:
	\begin{align*}
		S \models \gamma 
		&&\text{implies}&& 
		S_i \models \gamma_i  \quad \text{(for $i=1, \dots, k$)} \,;
	\end{align*}
	\item[Nice:] If for all $i$, $\gamma_i$ is provable from $S_i$, then $\gamma$ is provable from $S$. In symbols:
	\begin{align*}
		\frac{\; S_i \;}{\gamma_i}  \quad \text{(for $i=1, \dots, k$)} 
		&&\text{implies}&& 
		\frac{\; S \;}{\gamma} \,.
	\end{align*}
\end{description}
The reduction technique sketched above is similar to ``$\gamma$-reducibility'' of Adleman and Manders~\cite{AdlemanM79} (see also~\cite{ChungR89}), from which we borrowed our terminology.

We divide the reduction into four sets of moves, which are applied in the following order: (1) initialisation, (2) choice of domain, (3) reduction to canonical form, (4) reduction to proper canonical form; (5) reduction to polynomial form. 
The order of application is important for the correctness of the reduction.

\paragraph{Step 1 (Initialisation).}
The initialisation step has a single (deterministic) move
\begin{equation}
	(S,\phi \vdash \psi) \longrightarrow (S \cup \{p \vdash \phi, q \vdash \psi \}, p \vdash q)
	\tag{Init}
\end{equation}
where $p$, $q$ are fresh propositional variables not occurring $S$ and 
$\phi \vdash \psi$.
The intent of this move is to reduce the conclusion into a simplified canonical form 
---note that it is in polynomial. From this point onward, the conclusion $\gamma$ will be kept identical by every move. For this reason, abusing the notation, 
we define the next moves of the reduction without involving this component.

\paragraph{Step 2 (Choice of domain).}
The set of non-deterministic moves for this second step are given by 
\begin{align*}
	S &\longrightarrow S \cup \{\vdash |p|, q^2p \vdash 1 \}
	\tag{FP} \\
	S &\longrightarrow S[\bot/p]
	&& \text{(when $\vdash |p| \notin S$)}
	\tag{$\bot$} \\
	S &\longrightarrow S[0/p]
	&& \text{(when $q^2p \vdash 1 \notin S$)}
	\tag{0}
\end{align*}
where $p$, $q$ are propositional variables such that $p$ occurs in $S$ and $q$ is fresh in $S$.
In the above, $S[\phi/p]$ denotes substitution of a propositional variable $p$ for 
a formula $\phi$ in all the judgements of $S$.

Intuitively, (FP) non-derministically choose a propositional variable $p$ to be assumed finite and strictly positive ($x > 0$ iff exists $y$ such that $y^2 x \geq 1$). The moves ($\bot$) and ($0$) correspond,  respectively,
to deciding whether $p$ is infinite or zero. 
Observe that the rules can be applied in sequence until no more moves are available. 
The conditions of application imposed on the moves ($\bot$) and ($0$) makes sure that
the choice of domains for the propositional variables are coherent along a computation path. Due to the nondeterminism, all choices of domains are possible 
for the propositional variables in $S$.

\paragraph{Step 3 (Reduction to CF).} 
The reduction to canonical form consists of a single (determistic) move
\begin{align*}
	S &\longrightarrow S^{\sf cf} \cup \FF
	\tag{CF}
\end{align*}
ensuring that all judgements in $S$ are canonicalised. Although we already introduced the judgements of type $\vdash |p|$ in the previous step, taking the union with $\FF$ is necessary to prove that the move is nice. Indeed, after the canonicalisation, the judgement $\vdash |p|$ becomes trivially valid: $\vdash |p|^{\sf cf} = {\vdash 0}$.

\paragraph{Step 4 (Reduction to PCF).}
After the previous step of canonicalisation, the only judgements that are not in proper canonical form are either trivially valid ($\bot \vdash \phi$) or finitarily unsatisfiable ($\theta \vdash \bot$). The following moves keep their meaning when used as hypotheses 
but rewrite them in proper canonical form:
\begin{align*}
	S &\longrightarrow S \setminus \{ \bot \vdash \phi \mid \phi \in \pl \} 
	\tag{Valid} \\
	S &\longrightarrow \textit{Unsat}(S) 
	\tag{Unsat}
\end{align*}
where $\textit{Unsat}(S)$ is obtained from $S$ by replacing all the occurrences of
judgements of the form $\theta \vdash \bot$ with $1 \vdash 0$. Note that also 
$1 \vdash 0$ is unsatisfiable, but it is in polynomial form.

\paragraph{Step 5 (Reduction to Polynomial Form).} 
Recall that a formula is in polynomial form if it is in PCF and does not have occurrences
of $\lol$. The first requirement is guaranteed by the previous step. The moves given below
are designed to sequentially eliminate all the occurrences of $\lol$ inside a judgement:
{\small
	\begin{align*}
		\tag{$\ot$-L1}
		S \cup \{ \rho \ot (\vartheta \lol \gamma) \vdash \psi \} 
		&\longrightarrow \{ \rho \vdash \psi, \vartheta \vdash \gamma \} \cup S
		\\
		\tag{$\ot$-L2}
		S \cup  \{ \rho \ot (\vartheta \lol \gamma) \vdash \psi \} 
		&\longrightarrow \{ \rho\ot\gamma \vdash \psi\ot\vartheta, \gamma \vdash \vartheta \} \cup S
		\\
		\tag{$\ot$-R1}
		S \cup  \{ \phi \vdash \rho \ot (\vartheta \lol \gamma) \} 
		&\longrightarrow \{ \phi \vdash \rho, \vartheta \vdash \gamma \} \cup S
		\\
		\tag{$\ot$-R2}
		S \cup \{ \phi \vdash \rho \ot (\vartheta \lol \gamma) \}  
		&\longrightarrow \{ \phi\ot\vartheta \vdash \rho\ot\gamma, \gamma \vdash \vartheta \} \cup S
		\\
		\tag{m-L1}
		S \cup \{ \rho (\vartheta \lol \gamma) \vdash \psi \} 
		&\longrightarrow \{ \vdash \psi, \vartheta \vdash \gamma \} \cup S
		\\
		\tag{m-L2}
		S \cup \{ \rho (\vartheta \lol \gamma) \vdash \psi \} 
		&\longrightarrow \{ \rho\gamma \vdash \rho\vartheta \ot \psi, \gamma \vdash \vartheta \} \cup S
		\\
		\tag{m-R1}
		S \cup \{ \phi \vdash \rho (\vartheta \lol \gamma) \} 
		&\longrightarrow \{ \vartheta \vdash \gamma \} \cup S
		\\
		\tag{m-R2}
		S \cup \{ \phi \vdash \rho (\vartheta \lol \gamma) \}  
		&\longrightarrow \{ \phi\ot\rho\vartheta \vdash \rho\gamma, \gamma \vdash \vartheta \} \cup S
	\end{align*}
}
We assume the rules apply up to commutativity for both $\ot$ and multiplication.

\begin{proposition} \label{pro:soundnessMoves}
	The moves of the reduction are reliable and nice.
\end{proposition}

Now we are ready to prove our main theorem:
\begin{proof}[proof sketch of Theorem~\ref{thm:completenessPoly}]
Let $\gamma = \theta \vdash \vartheta$ be a polynomial judgement and 
$S = \{ \theta_1 \vdash \vartheta_1, \dots, \theta_n \vdash \vartheta_n \}$ a finite set of
polynomial judgements, all over the variables $p_1, \dots, p_m$. Assume that $S \models_{\preals} \gamma$.

$\sem{\theta}$, $\sem{\vartheta}$, $\sem{\theta_i}$, $\sem{\vartheta_i}$, 
$\sem{p_j}$ are polynomials in $\reals[X_1, \dots, X_m]$, where $X_j$ is the polynomial variable corresponding to $p_j$. 
The $\preals$-valued models of $S$ are the solutions of the following 
system of polynomial inequalities 
\begin{align*}
	\left\{\begin{aligned}
		\sem{\theta_i} - \sem{\vartheta_i} &\geq 0 &\text{(for $i = 1,\dots, n$)}\\
		\sem{p_j} &\geq 0 &\text{(for $j = 1,\dots, m$)}
	\end{aligned}\right.
\end{align*}
The hypothesis $S \models_{\preals} \gamma$ guarantees that all the solutions 
of the system above satisfy the inequality $\sem{\theta}-\sem{\vartheta} \geq 0$.
Let $\Sigma[X_1,\dots, X_m]$ denote the sum-of-squares polynomials in the 
variables $X_1, \dots, X_m$. By the Krivine-Stengle Positivstellensatz, there exist polynomials $h_1, h_2$ each of type 
\begin{equation*}
	\sum_{\alpha\in\{0,1\}^{n+m}} \sigma_\alpha 
	\Big(\prod_{i = 1}^n (\sem{\theta_i} - \sem{\vartheta_i})^{\alpha_i} \Big) 
	\Big(\prod_{j = 1}^m \sem{p_j}^{\alpha_{n+j}} \Big)
\end{equation*}
for some $\sigma_\alpha\in\Sigma[X_1,\dots,X_m]$ 
and integer $s$ such that
\begin{equation*}
	h_1\sem{\theta}=h_1 \sem{\vartheta}+(\sem{\theta}- \sem{\vartheta})^{2s}+h_2 \,.
\end{equation*}
The first step is to find formulas $\rho_1$, $\rho_2$ such that:
	\begin{equation}
		\label{eq01}
		\frac{S \quad \FF}{
			\vdash \rho_1\theta \lollol \rho_1\vartheta \ot (\vartheta \lol \theta)^{2s} \ot \rho_2} \,.
	\end{equation}
Note that, under the conditions of the finiteness $\FF$ of the propositional variables we can deduce both 
$\vdash |\rho_1|$ and $\vdash |\rho_2|$ (finiteness of $\rho_1$ and $\rho_2$).

From \eqref{eq01}, the main thesis follows by splitting the proof into three sub-cases: $s = 0$; $s > 0$; and $s < 0$. The first case is the simplest, and the proof follows by syntactical manipulation. The other two cases are significantly more involved, both requiring the use of the weak form of Krivine-Stengle Positivstellensatz.
\qed\end{proof}

Theorem~\ref{thm:completeness} has two simple corollaries regarding approximated completeness that could be relevant from a computational perspective.

\begin{corollary}\label{approxlemma}
Let $S$ be a finite set of judgements and $\phi \vdash \psi$ a polynomial judgement. 
If there exists $\e>0$ such that $S \models \phi \vdash \psi \ot \e$, then 
  $\frac{S}{\phi\vdash\psi}$.
\end{corollary}

\begin{corollary}[Approximated Completennes]
\label{approxcompleteness}
Let $S$ be a finite set of judgements. If a judgement $\phi\vdash\psi$ is a semantic consequence of $S$ in $\pl$, \ie, $S \models \phi\vdash\psi$, then for any $\delta>0$, 
$\frac{S}{\phi\ot\delta\vdash\psi}$.
\end{corollary}


\section{Complexity Results} 
\label{sec:complexity}

In this section, we present complexity bounds for standard decidability problems, specifically, satisfiability for a finite set of judgements and semantic consequence. We provide these results for both $\pl$ and its fragment $\al$.

The \emph{satisfiability problem} for a finite set of judgments $S$ is the problem of determining whether there exists a model $\M$ such that $\M \models S$. Next, we investigate its complexity for $\al$ and $\pl$, simultaneously.

\smallskip
Below, we describe a non-deterministic procedure for transforming a set of judgments in
proper canonical form (PCF) into a $\preals$-equi-satisfiable set of judgements in polynomial form (resp.\ affine form).

The algorithm is defined by a finite set of nondeterministic moves 
\begin{equation*}
  S \longrightarrow S_i  \qquad \text{(for $i=1,\dots,k$)}
\end{equation*}
between set of judgements $S$, $S_i$, such that 
\begin{itemize}[topsep=1ex]
  \item \emph{(Form preservation)} if $S$ is in PCF, so is $S_i$;
  \item \emph{(Soundness)} $S$ is $\preals$-satisfiable iff one of the $S_i$ is $\preals$-satisfiable.
\end{itemize}
Each move $S \longrightarrow S'$ eliminates an occurrence of $\lol$ from the judgements in $S$. The algorithm repeatedly applies the moves until it reaches a terminal configuration, which, by design, must be a set of judgements in polynomial form.


We obtain efficient moves by introducing fresh variables and allowing duplications only for these. As each move removes exactly one occurrence of $\lol$, the duplication of a propositional variable does not duplicate the number of moves.

The set of efficient moves for the elimination of $\lol$ are
{\small
\begin{align*}
  \tag{$\ot$-L1*}
  S \cup \{ \rho \ot (\vartheta \lol \theta) \vdash \psi \} 
  &\to S \cup \{ \rho \vdash \psi, \vartheta \vdash \theta \}
  \\
  \tag{$\ot$-L2*}
  S \cup \{ \rho \ot (\vartheta \lol \theta) \vdash \psi \} 
  &\to S \cup \{
  	\rho\ot p \vdash \psi\ot q, 
	\theta \vdash p, p \vdash q, q \vdash \vartheta
  \}
  \\
  \tag{$\ot$-R*}
  S \cup \{ \phi \vdash \rho \ot (\vartheta \lol \theta) \} 
  &\to S \cup \{ \phi \vdash \rho\ot p, p \ot\vartheta \vdash \theta \}
  \\
  \tag{m-L1*}
  S \cup \{ \rho (\vartheta \lol \theta) \vdash \psi \} 
  &\to S \cup \{ \vdash \psi, \vartheta \vdash \theta \}
  \\
  \tag{m-L2*}
  S \cup \{ \rho (\vartheta \lol \theta) \vdash \psi \} 
  &\to S \cup \{
  	pq \vdash pr \ot \psi, \rho \vdash p, 
	\theta \vdash q, q \vdash r, r \vdash \vartheta
  \}
  \\
  \tag{m-R*}
  S \cup \{ \phi \vdash \rho (\vartheta \lol \gamma) \} 
  &\to S \cup \{ \phi \vdash \rho p, p \ot\vartheta \vdash \gamma \}
\end{align*}
}
where $p,q,r$ are propositional variables fresh in $\phi, \psi, \rho, \vartheta, \theta$, and $S$.

\begin{proposition} \label{prop:EffLOLelimination}
The moves for the elimination of $\lol$ are sound.
\end{proposition}

Based on this algorithm, we can propose a non-deterministic procedure 
for deciding the satisfiability of a finite set of judgements both in $\al$ and $\pl$.

We start by establishing the complexity of deciding satisfiability in $\al$.
\begin{theorem} \label{thm:SatALNP-compl}
Satisfiability in $\al$ is NP-complete.
\end{theorem}
NP-hardness follows from a reduction of SAT and NP-membership via the non-deterministic procedure described above, which turns satisfiability in $\al$ into the feasibility checking of affine inequalities (solved  
by using Ellipsoid method~\cite{Khachiyan1980}).

Next, we state the complexity of deciding satisfiability in $\pl$. 
Observe that as $\al$ is a sublogic of $\pl$, by Theorem~\ref{thm:SatALNP-compl}, 
the problem is at least NP-hard. 
\begin{theorem}\label{SATpl-PSPACE}
Satisfiability in $\pl$ is in PSPACE.
\end{theorem}
The proof is via the same non-deterministic procedure, which reduces the problem to the
satisfiability of a formula in the existential theory of the real numbers, which can be 
decided in PSPACE~\cite{Canny88}.


The \emph{semantical consequence} from a set of judgements $S$ is the problem of determining whether a model $\M$ that satisfies all the judgements of $S$ does also satisfy a given judgement $\gamma$, called consequent.

Next, we state the complexity of this problem both in $\al$ and $\pl$.

\begin{theorem} \label{thm:cosequenceAPisCoNP}
Semantical consequence in $\al$ is co-NP complete.
\end{theorem}
The co-NP hardness follows by a linear-time reduction from the tautology
problem for Boolean propositional logic. Membership in co-NP follows via the 
non-deterministic procedure explained at the beginning of the section,
which reduces the problem to the infeasibility of linear programs~\cite{Khachiyan1980}.

Finally, we give complexity bounds also for the corresponding problem in $\pl$.
\begin{theorem}\label{SCpl-PSPACE}
Semantical consequence in $\pl$ is in PSPACE.
\end{theorem}
Membership in PSPACE follows via the non-deterministic procedure explained 
at the beginning of the section, which reduces the problem to that of deciding the 
satisfiability of a formula in the existential theory of the reals~\cite{Canny88}.

\section{Conclusions}
\label{sec:concl}

This paper develops and studies Polynomial Lawvere logic ($\pl$), a logic based on the Lawvere quantale extended with multiplication, whose monoidal structure interacts with that on the tensor as a semiring. This logic is well-suited for encoding quantitative reasoning principles naturally, as demonstrated in Section~\ref{sec:applications}. 

We propose a deduction system for $\pl$ and show that, while the logic is generally incomplete, finitely axiomatized theories are complete. The core of the completeness proof draws on results from algebraic geometry, specifically the Krivine-Stengle Positivstellensatz. 
The use of such results in the completeness proof provides compelling evidence of the deep connection between arithmetic and logical reasoning.

Additionally, we present new complexity results for both $\pl$ and its affine fragment ($\al$). 
We demonstrate that the satisfiability of a finite set of judgements is NP-complete in $\al$ and in PSPACE for $\pl$; and that deciding the semantical consequence from a finite set of judgements is co-NP complete in $\al$ and in PSPACE for $\pl$.

Finally, building on the Weierstrass approximation theorem, which states that continuous real-valued functions on compact subsets can be approximated arbitrarily well by polynomials, one might consider developing an approximation theory grounded in $\pl$. This direction
is left for future work.

%
%
%
\bibliographystyle{splncs04}
\bibliography{main}

\newpage

\begin{appendix}

\section{Appendix} \label{appendix}
Below, we provide detailed proofs of the results stated or implicitly used in the main text, which were omitted for brevity.
	
	\subsection{The proofs of the main results}
	
We start this section by stating a useful lemma that we will use extensively in the proofs below.
\begin{lemma}[Disjunction Deduction Lemma] \label{totality}
Let $\gamma$ be a judgement, $S$ a finite set of judgements and 
$\phi$, $\psi$ formulas in $\pl$. Then,
\begin{align*}
\Big( \;
\dfrac{}{\vdash \phi \lor \psi} \,, \;
\dfrac{S \quad \vdash\phi}{\gamma}
\text{ and }
\dfrac{S \quad \vdash\psi}{\gamma}
\; \Big)
&&\text{ implies }&&
\dfrac{~S~}{\gamma} \,.
\end{align*}
\end{lemma}
\begin{proof}
It is an equivalent formulation of the totality lemma~\cite{Bacci2023} (Lemma 4.6).
\end{proof}

\begin{proof}[of Lemma~\ref{polid}]
Two polynomials are equal iff their monomials have same coefficients. Thus proving ${\vdash\theta\lollol\vartheta}$ from $\sem{\theta} = \sem{\vartheta}$ is done by repeatedly applying (\textsc{d1}), (\textsc{d3}) and the associativity and commutativity for $\ot$ and multiplication. For the converse, from the soundness of ${\vdash\theta\lollol\vartheta}$ follows that $\sem{\theta}|_{[0,\infty)^n} = \sem{\vartheta}|_{[0,\infty)^n}$. Since $\sem{\theta}$, $\sem{\vartheta}$ are polynomials, this implies $\sem{\theta} = \sem{\vartheta}$. 
\qed\end{proof}

\begin{proof}[of Corollary \ref{polidext}]
	Since $\sem{\theta} = \sem{\vartheta}$, by Lemma~\ref{polid}, $\vdash\phi\lollol\psi$ is
	provable in $\pl$. Then, from this and the two hypothesis 
	\begin{align*}
		\frac{S}{\phi\lollol\psi \vdash \theta\lollol\vartheta} 
		&&\text{and}&&
		\frac{S}{\theta\lollol\vartheta \vdash \phi\lollol\psi} 
	\end{align*}
	we get $\displaystyle\frac{S}{\vdash\phi \lollol \psi}$.
	\qed\end{proof}	
	
\begin{proof}[of Proposition \ref{prop:cfOFcf}]
	Provability of $\phi^{\sf cf} \vdash \phi$ follows by Propositions~\ref{prop:cfSemantic} and \ref{prop:cfOFcf}. Provability of $\phi^{\sf cf} \vdash \phi$ is shown by induction on the 
	formula by applying monotonicity of logical connectives. The validity of both judgements is a consequence of Theorem~\ref{soundness}.
	The second part is proven by by an easy induction on the structure of formulas in PCF with the base case given by definition: $\bot = \bot^{\sf cf}$.
	\qed\end{proof}


\begin{proof}[of Theorem \ref{thm:completeness}]
We present here how the statement of Theorem \ref{thm:completeness} can be reduced to the polynomial version presented in the paper.
\end{proof}

\begin{proof}[of Proposition~\ref{pro:soundnessMoves}]
Reliability is easy to check in each case of the reduction, as the moves were designed
exactly to preserve this invariant. 
As for niceness, we consider each step separately. Step 1 follows by substitution as if 
we assume that the following is provable 
$$\dfrac{S \quad p \vdash \phi \quad q \vdash \psi}{p \vdash q} \,,$$ so is any substitution instance of it. Indeed, $\dfrac{S}{\phi \vdash \psi}$ is obtainded by substituting $\phi$ for $p$ and $\psi$ for $q$. Step 2 follows by Lemma~\ref{totality}. Step 3 follows by the fact that after Step~2, $S$ contains judgements of the form $q_i^2p_i \vdash 1$ for each propositional variable $p_i$ in $S$. Since it is possible to deduce $\vdash p_i > 0$ from $q_i^2p_i \vdash 1$ the result follows by Proposition~\ref{prop:fpProvabilityCF}. 
Step~4 is clear. Step 5 follows by the rule ($\ot_2$), ($\lol_1$), ($\ot_1$) and from Lemma~\ref{totality} by exploiting the dichotomies (\textsc{wem}) and (\textsc{tot}).
\qed\end{proof}


\begin{proof}[detailed proof of Theorem~\ref{thm:completenessPoly}]
Let $\gamma = \theta \vdash \vartheta$ be a polynomial judgement and 
$S = \{ \theta_1 \vdash \vartheta_1, \dots, \theta_n \vdash \vartheta_n \}$ a finite set of
polynomial judgements, all over the variables $p_1, \dots, p_m$. Assume that $S \models_{\preals} \gamma$.

$\sem{\theta}$, $\sem{\vartheta}$, $\sem{\theta_i}$, $\sem{\vartheta_i}$, 
$\sem{p_j}$ are polynomials in $\reals[X_1, \dots, X_m]$, where $X_j$ is the polynomial variable corresponding to $p_j$. 
The $\preals$-valued models of $S$ are the solutions of the following 
system of polynomial inequalities 
\begin{align}
	\label{app:Sys0}
	\left\{\begin{aligned}
		\sem{\theta_i} - \sem{\vartheta_i} &\geq 0 &\text{(for $i = 1,\dots, n$)}\\
		\sem{p_j} &\geq 0 &\text{(for $j = 1,\dots, m$)}
	\end{aligned}\right.
\end{align}
The hypothesis $S \models_{\preals} \gamma$ guarantees that all the solutions 
of the system above satisfy the inequality $\sem{\theta}-\sem{\vartheta} \geq 0$.
Let $\Sigma[X_1,\dots, X_m]$ denote the sum-of-squares polynomials in the 
variables $X_1, \dots, X_m$. By the Krivine-Stengle Positivstellensatz, there exist polynomials $h_1, h_2$ each of type 
\begin{equation*}
	\sum_{\alpha\in\{0,1\}^{n+m}} \sigma_\alpha 
	\Big(\prod_{i = 1}^n (\sem{\theta_i} - \sem{\vartheta_i})^{\alpha_i} \Big) 
	\Big(\prod_{j = 1}^m \sem{p_j}^{\alpha_{n+j}} \Big)
\end{equation*}
for some $\sigma_\alpha\in\Sigma[X_1,\dots,X_m]$ 
and integer $s$ such that
\begin{equation}
	h_1\sem{\theta}=h_1 \sem{\vartheta}+(\sem{\theta}- \sem{\vartheta})^{2s}+h_2 \,.
	\label{app:Null0}
\end{equation}

\begin{description}[fullwidth]
	\item[Case 1: $s \geq 0$.]
	Our goal is to find formulas $\rho_1$, $\rho_2$ such that:
	\begin{equation}
		\label{app:eq01}
		\frac{S \quad \FF}{
			\vdash \rho_1\theta \lollol \rho_1\vartheta \ot (\vartheta \lol \theta)^{2s} \ot \rho_2} \,.
	\end{equation}
	Equation~\eqref{app:Null0} suggests that $\rho_1$, $\rho_2$ are the formulas in $\pl$
	(not necessary in polynomial form!) corresponding to $h_1$, $h_2$, respectively,
	in the sense that $\M(\rho_i) = h_i(\M(p_1), \dots, \M(p_n))$, for all $\preals$-valued models $\M$.
	Next, we show how to obtain such formulas.
	
	Identify in each side of equation~\eqref{app:Null0} all the occurrences of $\sem{\theta} - \sem{\vartheta}$ and $\sem{\theta_i} - \sem{\vartheta_i}$ with $\vartheta \lol \theta$ and $\vartheta_i \lol \theta_i$, respectively, and all occurrences of $\sem{p_j}$ with $p_j$. Then, replace $+$ by $\ot$ and product by multiplication between formulas. 
	After applying this syntactical rewriting process we obtain two
	formulas:
	\begin{align*}
		\rho_1\theta 
		&&\text{and}&& 
		\rho_1\vartheta \ot (\vartheta \lol \theta)^{2s} \ot \rho_2 \,.
	\end{align*}
	corresponding to the left- and right-hand side of \eqref{app:Null0}, respectively.
	
	Transform equation~\eqref{app:Null0} in an equality between 
	two polynomials that contains no occurrences of $-$ (subtraction), say $h_3 = h_4$.
	The same calculation can be reproduced in $\pl$ by repeatedly applying the following  inference valid in $\pl$
	\begin{equation*}
		\infrule{\phi_1\vdash\phi_2 & \FF}{ \vdash 
			\big( \psi_1\lollol\psi_2\ot\psi_3(\phi_2\lol\phi_1) \big) 
			\lollol 
			\big( \psi_1\ot\psi_3\phi_2\lollol\psi_2\ot\psi_3\phi_1 \big)
		}
	\end{equation*}
	where the r{\^o}le of $\phi_1 \vdash \phi_2$ will be taken either by $\theta \vdash \vartheta$, 
	or $\theta_i \vdash \psi_i$ ($i=1,\dots,n$), which all belong in $S$. This allows us to eliminate the occurrences of $\lol$. 
	
	We get two polynomial formulas, $\rho_3$, $\rho_4$ s.t. 
	$\sem{\rho_3} = h_3$ and $\sem{\rho_4} = h_4$. Moreover,
	\begin{gather*}
		\frac{S \quad \FF}{\rho_1\theta \lollol \rho_1\vartheta \ot (\vartheta \lol \theta)^{2s} \ot \rho_2 
			\vdash \rho_3 \lollol \rho_4} 
		\intertext{and}
		\frac{S \quad \FF}{\rho_3 \lollol \rho_4 \vdash \rho_1\theta \lollol \rho_1\vartheta \ot (\vartheta \lol \theta)^{2s} \ot \rho_2 } 
	\end{gather*}
Corollary~\ref{polidext} gives us \eqref{app:eq01} as required.
	
	\begin{description}[fullwidth]
		\item[Subcase 1.1: $s > 0$.] 
		From \eqref{app:eq01}, we get 
		\begin{equation*}
			\frac{S \quad \FF}{ \rho_1\theta\vdash\rho_1\vartheta\ot(\vartheta\lol\theta)^{2s}\ot\rho_2} \,,
		\end{equation*}
		implying $\dfrac{S \quad \FF}{\rho_1\theta\vdash\rho_1\vartheta}$ and eventually 
		$\dfrac{S \quad \FF}{\vdash\rho_1(\theta\lol\vartheta)}$, where from we get
		\begin{equation}\label{eq02}
			\frac{S \quad \FF}{\vdash\rho_1\lor(\theta\lol\vartheta)}
		\end{equation}
		We aim to prove $\dfrac{S \quad \FF}{\vdash\theta\lol\vartheta}$ using Lemma~\ref{totality} 
		and for this, we will need to prove 
		\begin{align*}
			\textsc{(a)} \;\; \frac{S \quad \FF  \quad \rho_1\vdash\theta\lol\vartheta}{\vdash\theta\lol\vartheta} 
			&& \text{ and }&&
			\textsc{(b)} \;\; \frac{S \quad \FF \quad \theta\lol\vartheta\vdash\rho_1}{\vdash\theta\lol\vartheta} \,.
		\end{align*}
		\textsc{(A)} follows directly from~\eqref{eq02}. 
		For \textsc{(b)}, note that from~\eqref{eq02}, we get
		$\dfrac{S \quad \FF \quad \theta\lol\vartheta\vdash\rho_1}{\vdash \rho_1}$
		and applying this in \eqref{app:eq01}, we obtain  
		$\dfrac{S \quad \FF \quad \theta\lol\vartheta\vdash\rho_1}{\vdash(\vartheta\lol\theta)^{2s}\ot \rho_2}$ 
		and then $\dfrac{S \quad \FF  \quad \theta\lol\vartheta\vdash\rho_1}{\vdash(\vartheta\lol\theta)^{2s}}$. 
		
		Because $s>0$, this gives us
		$\dfrac{S \quad \FF  \quad \theta\lol\vartheta\vdash\rho_1}{\vartheta\vdash\theta}$.
		Using the soundness we obtain that $S \cup \{ \theta\lol\vartheta \vdash \rho_1\} \models_{\preals} \vartheta\vdash\theta$. By hypothesis, $S \models_{\preals} \theta \vdash \vartheta$, hence 
		$$S \cup \{ \theta\lol\vartheta \vdash \rho_1 \} \models_{\preals} \; \vdash \vartheta \lollol \theta \,.$$
		By identifying the finitary models of $S \cup \{ \theta \lol \vartheta \}$ with the solutions
		of the following system of polynomial inequalities
		\begin{align*}
			\left\{\begin{aligned}
				\sem{\theta_i} - \sem{\vartheta_i} &\geq 0 &\text{(for $i = 1,\dots, n$)}\\
				\sem{p_j} &\geq 0 &\text{(for $j = 1,\dots, m$)}\\
				(\sem{\vartheta} - \sem{\theta}) - h_1 &\geq 0
			\end{aligned}\right.
		\end{align*}
		$S \cup \{ \theta\lol\vartheta \vdash \rho_1 \} \models_{\preals} \vdash \vartheta \lollol \theta$ 
		implies that the solutions of the system above satisfy the equality $\sem{\vartheta} - \sem{\theta} = 0$. Applying the weak Positivstellensatz we get that there exists an 
		integer $r > 0$ and one inequality $h \geq 0$ in the system above so that 
		$h + (\sem{\vartheta} - \sem{\theta})^{2r} = 0$.
		Assuming $\rho$ represents $h$, this fact translates into
		\begin{equation*}
			\frac{S \quad \FF \quad \theta\lol\vartheta \vdash \rho_1}{\rho \ot (\theta \lol \vartheta )^{2r}}
		\end{equation*}
		which implies \textsc{(B)} as required.
		
		\item[Subcase 1.2: $s = 0$.]
		The instantiation of \eqref{app:eq01} for $s = 0$, implies that 
		\begin{equation}
			\label{eq05}
			\frac{S \quad \FF}{
				\rho_1\theta \vdash \rho_1\vartheta \ot 1 \ot \rho_2} \,.
		\end{equation}
		From the above we obtain $\dfrac{S \quad \FF}{\rho_1\theta \vdash \rho_1\vartheta}$, which is provably equivalent to $\dfrac{S \quad \FF}{\vdash \rho_1(\theta \lol \vartheta)}$ and further implies 
		\begin{equation}\label{eq06}
			\dfrac{S \quad \FF }{\vdash \rho_1 \lor (\theta \lol \vartheta)} \,.
		\end{equation} 
		We will again use Lemma~\ref{totality} to obtain 
		the desired result. For this, it is sufficient to prove the inferences:
		\begin{align*}
			\textsc{(a)} \;\; \frac{S \quad \FF \quad \rho_1\vdash\theta\lol\vartheta}{\vdash\theta\lol\vartheta} 
			&& \text{ and }&&
			\textsc{(b)} \;\; \frac{S \quad \FF \quad \theta\lol\vartheta\vdash\rho_1}{\vdash\theta\lol\vartheta} \,.
		\end{align*}
		\textsc{(A)} follows directly from~\eqref{eq02}. As for \textsc{(B)}, from \eqref{eq06}, we
		get $\dfrac{S  \quad \FF  \quad \theta\lol\vartheta\vdash\rho_1}{\vdash \rho_1}$, which applied in \eqref{eq05} gives $\dfrac{S  \quad \FF  \quad \theta\lol\vartheta\vdash\rho_1}{\vdash 1}$, 
		implying $\dfrac{S  \quad \FF  \quad \theta\lol\vartheta\vdash\rho_1}{\vdash \bot}$ and next \textsc{(b)}.

	\end{description}
	
	\item[Case 2: $s = -r < 0$ with $r>0$.]
	We have $\sem{\theta}- \sem{\vartheta} \neq 0$ (otherwise, $(\sem{\theta}- \sem{\vartheta})^{2s}$ is not well defined in equation \eqref{app:Null0}) and \eqref{app:Null0} is equivalent to
	\begin{equation}
		h_1(\sem{\theta}- \sem{\vartheta})^{2r+1}=1 + h_2(\sem{\theta} - \sem{\vartheta})^{2r} \,.
	\end{equation}
	From $S \models_{\preals} \gamma$ and by knowing that $\sem{\theta} - \sem{\vartheta} \neq 0$ 
	over all of the reals, we get that all the solutions of the system of polynomial inequalities \eqref{app:Sys0} necessarily satisfy the \emph{strict} inequality $\sem{\theta} - \sem{\vartheta} > 0$. Then, by the Krivine-Stengle Positivstellensatz, there exist polynomials $h_5, h_6$ each of type 
	\begin{equation*}
		\sum_{\alpha\in\{0,1\}^{n+m}} \sigma_\alpha 
		\Big(\prod_{i = 1}^n (\sem{\theta_i} - \sem{\vartheta_i})^{\alpha_i} \Big) 
		\Big(\prod_{j = 1}^m \sem{p_j}^{\alpha_{n+j}} \Big)
	\end{equation*}
	for some $\sigma_\alpha\in\Sigma[X_1,\dots,X_m]$ such that
	\begin{equation}
		h_5\sem{\theta}=h_5 \sem{\vartheta}+1+h_6 \,.
		\label{Null2}
	\end{equation}
	As before, there will be two formulas $\rho_5$, $\rho_6$ corresponding to the
	the polynomials $h_5$, $h_6$, respectively. This translate in $\pl$ as
	\begin{equation*}
		\frac{S  \quad \FF }{
			\rho_5\theta \vdash \rho_1\vartheta \ot 1 \ot \rho_6} \,.
	\end{equation*}
	From here, the thesis follows as in subcase 1.2.
	\qed
\end{description}

\end{proof}


	
\begin{proof}[of Theorem \ref{thm:SatALNP-compl}]
	We first prove the membership in NP.
	
	Let $S$ be a finite set of judgements in $\al$ over the variables 
	$p_1, \dots, p_n$. A non-deterministic procedure for checking the satisfiability of $S$
	is as follows:
	\begin{enumerate}
		\item guess a subset of the propositional variables;
		\item replace every other variable by $\bot$ in the judgements of $S$;
		\item canonicalise the judgements;
		\begin{itemize}
			\item If there is a finitely unsatisfiable judgement of the form $\theta \vdash \bot$, 
			terminate negatively;
			\item Otherwise, remove all the trivially satisfiable judgements of the form
			$\bot \vdash \theta$. 
		\end{itemize}
		\noindent
		At this point, all the judgments left are in proper canonical form.
		\item Proceed with the non-deterministic algorithm for the elimination of $\lol$.
		At the end of this procedure we are left with a set of judgements in affine form;
		\item Finite satisfiability for affine set of judgements is determined by 
		a polynomial-time reduction to the feasibility of linear programs.
	\end{enumerate}
	As the starting set of judgements is finite, termination of the non-deterministic 
	procedure is evident. 
	
	Correctness follows by Lemma~\ref{totality} and
	Proposition~\ref{prop:EffLOLelimination} and the fact that
	all judgements in $\al$ are $\preals$-provably equivalent to their canonical forms, \ie, for all formulas $\phi$,$\psi$ in $\al$
	\begin{align*}
		\frac{\phi \vdash \psi \quad \FF}{\phi^{\sf cf} \vdash \psi^{\sf cf}}
		&& \text{and} &&
		\frac{\phi^{\sf cf} \vdash \psi^{\sf cf} \quad \FF}{\phi \vdash \psi} \,.
	\end{align*}
	(The proof of this last claim is similar to that of Proposition~\ref{prop:fpProvabilityCF}).
	
	As for the complexity, Steps 2--3 all take linear time, as well as each non-deterministic computation branch in Step 4. Step 5 can be
	computed in polynomial-time in the size%
	\footnote{The size of a formula is the sum of the number of logical connectives and propositional variables used plus
		the total number of bits in a binary representation of the coefficients used in scalar 
		multiplications. The size of a set of judgements the overall sum of the sizes of the 
		formulas in it.} of the affine set of judgements obtained after Step 6, by using the 
	Ellipsoid method by Khachiyan~\cite{Khachiyan1980}.
	
	Since each move for the elimination of $\lol$ adds only a constant number of formulas, each of constant size, the size of the set of judgements after Step 6 increased at most linearly 
	in the size of the input. Thus, each non-deterministic computation branch takes at most polynomial time
	in the size of the input.
	
	NP-hardness follows via a reduction from SAT. Let $\varphi$ be a Boolean
	formula with propositional variables in $\Prop$ built using only connectives $\lnot$, $\land$, and $\lor$. Define $\varphi^{\lnot\lnot}$ as the formula in $\al$ obtained by adding double negation ($\lnot\lnot$) before any propositional variable in $\varphi$.
	Then, $\varphi$ is Boolean-satisfiable if and only if $\vdash \varphi^{\lnot\lnot}$ is $\al$-satisfiable (as $\models_\M \lnot\lnot p$ iff $\M(p) < \infty$, for any $p \in \Prop$). 
	\qed\end{proof}	
	
	\begin{proof}[of Theorem \ref{SATpl-PSPACE}]
		Let $S$ be a finite set of judgements in $\pl$ over the propositional variables 
		$p_1, \dots, p_n$. We give a $\gamma$-reduction, in the sense of Andleman and Manders~\cite{AdlemanM79},
		to the existential theory of the
		real numbers, which can be decided in PSPACE~\cite{Canny88}. The non-deterministic reduction procedure is described in steps:
		
		The Steps 1, 2 and 3 are identical to the ones in the proof of 
		Theorem \ref{thm:SatALNP-compl}.
		
		Step 4: Let $q_1, \dots, q_n$ be fresh propositional variables. 
		Add the proper canonical judgments $p_iq_i^2 \vdash 1$ (for $i = 1, \dots, n$).
		
		Step 5: Proceed with the non-deterministic elimination of $\lol$.
		At the end of this procedure we are left with a set of judgements in polynomial form.
		
		Finite satisfiability for polynomial set of judgements is equivalent to feasibility of a set of inequalities between polynomials which is expressible in the existential theory of the real numbers.
		
		Termination of the algorithm is evident. Correctness of the reduction follows by Lemma~\ref{totality} and \ref{prop:EffLOLelimination}. 
		As for the complexity, Steps 2--3 all take linear time, as well as each non-deterministic computation branch in Step 4. So the proposed is in NP. 
		The thesis follows as $\text{NP} \subseteq \text{PSPACE}$.
		\qed\end{proof}
		
		\begin{proof}[of Theorem \ref{thm:cosequenceAPisCoNP}]
			Let $S$ be a finite set of affine judgements and $\gamma = \phi \vdash \psi$ an affine judgement, all using only the variables $p_1, \dots, p_n$.
			A non-deterministic procedure for checking
			$S \not\models \gamma$ follows the steps of the non-deterministic reduction proposed 
			for the proof of finite completeness above.
			As observed previously, we need to pay attention to the efficiency of the moves.
			Thus, as done for SAT complexity, we replace the moves in Step 5 
			(reduction to polynomial form) with the set of moves for the efficient elimination of $\lol$.
			Note also that, in Step 3 (reduction to CF), there is no need to introduce $\FF$ in the
			hypotheses, as long as we remember to check for finite unsatisfiability from that 
			point onward. Similarly, in Step 2 (choice of domain) we don't need to add judgements 
			of type $q^2p \vdash 1$, as 
			all judgements in $\al$ are $\preals$-provably equivalent to their canonical forms (the proof is similar to that of Proposition~\ref{prop:fpProvabilityCF}).
			
			Thus, at the end of the non-deterministic procedure, we are left with sets 
			$S_i = \{ \theta_{i1} \vdash \vartheta_{i1}, \dots, \theta_{in} \vdash \vartheta_{in} \}$ of 
			affine judgements with corresponding conclusions $\gamma_i= \phi_i \vdash \psi_i$, 
			which is also affine, such that
			$S \not\models \gamma$ iff  $S_i \not\models_{\preals} \gamma_i$, for some $i$.
			
			Determining $S_i \not\models_{\preals} \gamma_i$ is equivalent to the
			feasibility of the following system of linear inequalities 
			\begin{align*}
				\left\{\begin{aligned}
					\sem{\theta_{ik}} &\geq \sem{\vartheta_{ik}} &\text{(for $k = 1,\dots, n$)}\\
					\sem{p_j} &\geq 0 &\text{(for $j = 1,\dots, m$)}\\
					\sem{\phi_i} &< \sem{\psi_i} 
				\end{aligned}\right.
			\end{align*}
			which can be checked via a polynomial-time reduction to the infeasibility of linear programs 
			(in fact, it corresponding dual linear program) which can be done in polynomial-time by employing the Ellipsoid method by Khachiyan~\cite{Khachiyan1980}.
			
			From similar considerations done in the proof of Theorem~\ref{thm:SatALNP-compl} , 
			we can show that each non-deterministic computation branch takes at most polynomial time
			in the size of the input. Thus the problem of semantical consequence is in co-NP.
			
			Since we can encode Boolean propositional logic in $\al$, the hardness
			follows by a linear-time reduction from the tautology problem for Boolean propositional logic.
			\qed\end{proof}
		
		\begin{proof}[of Theorem \ref{SCpl-PSPACE}]
			Let $S$ be a finite set of judgements and $\gamma$ a judgement over the propositional variables $p_1, \dots, p_n$. The reduction proposed in the finite completeness proof ---but modified by replacing the moves of elimination of $\lol$ with the efficient ones from SAT complexity proof --- is a $\gamma$-reduction in the sense of Andleman and Manders~\cite{AdlemanM79}. It reduces the problem of semantical consequence in $\pl$ 
			to the satisfiability of a formula of the language of the existential theory of 
			the real numbers (details are in the proof of completeness).
			Termination of the non-deterministic algorithm is evident. Correctness follows from Propositions~\ref{pro:soundnessMoves} and~\ref{prop:EffLOLelimination}.
			
			As for the complexity, all the steps of the reduction take linear time. 
			So the proposed procedure is in NP. The thesis follows as satisfiability of a formula 
			in the existential theory of the reals is in PSPACE and $\text{NP} \subseteq \text{PSPACE}$.
			\qed\end{proof}

	
	\subsection{Useful lemmas of $\pl$} \label{app:useful}
	
	\begin{lemma}\label{l1}
        The following are derivable proof rules in $\pl$
		\begin{align*}
			\begin{array}{c}
				\begin{aligned}
					\infrule[1.]{}{\phi\land\psi\vdash\phi}
					&&&
					\infrule[2.]{}{\phi\land\psi\vdash\psi}
					&&&
					\infrule[3.]{\vdash\phi && \vdash\psi}{\vdash\phi\land\psi}
					\\[2ex]
					\infrule[4.]{\phi\vdash\psi}{\vdash\phi\lollol\phi\land\psi}
					&&&
					\infrule[5.]{\psi\vdash\phi}{\vdash\psi\lollol\phi\land\psi}
					&&&
					\infrule[6.]{\phi\vdash\psi}{\rho\lol\phi\vdash\rho\lol\psi}
					\\[2ex]
					\infrule[7.]{}{\phi\vdash\phi\lor\psi}
					&&&
					\infrule[8.]{}{\psi\vdash\phi\lor\psi}
					&&&
					\infrule[9.]{\phi\vdash\psi}{\vdash\phi\lor\psi\lollol\psi}		\end{aligned}
			\end{array}
		\end{align*}
		
	\end{lemma}
	
	\begin{proof} We prove each derivation separately:
 
		1. $$\infrule[def $\land$]{\infrule[weak,$\ot_1$]{\infrule[id]{}{\vdash\phi\lol\phi}}{\phi\ot(\phi\lol\psi)\vdash\phi}}{\phi\land\psi\vdash\phi}$$
		
		2. $$\infrule[def $\land$]{\infrule[$\ot_2$]{\infrule[id]{}{\phi\lol\psi\vdash\phi\lol\psi}}{\phi\ot(\phi\lol\psi)\vdash\psi}}{\phi\land\psi\vdash\psi}$$
		
		3. $$\infrule[def $\land$]{\infrule{\vdash\phi && \infrule[$\ot_2$]{\infrule[weak]{\vdash\psi}{\phi\vdash\psi}}{\vdash\phi\lol\psi}}{\vdash\phi\ot(\phi\lol\psi)}}{\vdash\phi\land\psi}$$
		
		4. $$\infrule[Lemma\ref{l1}.3]{\infrule[$\ot_2$]{\infrule{\infrule[$\ot_2$]{\phi\vdash\psi}{\vdash\phi\lol\psi}}{\phi\vdash\phi\ot(\phi\lol\psi)}}{\vdash\phi\lol(\phi\ot(\phi\lol\psi))}&&
			\infrule[$\ot_2$]{\infrule{\infrule[top]{}{\phi\lol\psi\vdash\top}}{\phi\ot(\phi\lol\psi)\vdash\phi}}{\vdash(\phi\ot(\phi\lol\psi))\lol\phi}}{\infrule[def $\land$]{\vdash\phi\lollol(\phi\ot(\phi\lol\psi))}{\vdash\phi\lollol\phi\land\psi}}$$
		
		5. $$\infrule[def $\land$]{\infrule[$\lol_1$]{\infrule[id]{}{\phi\lol\psi\vdash\phi\lol\psi} && \psi\vdash\phi}{\psi\vdash\phi\ot(\phi\lol\psi)} &&
			\infrule[$\ot_2$]{\infrule[id]{}{\phi\lol\psi\vdash\phi\lol\psi}}{\phi\ot(\phi\lol\psi)\vdash\psi} }{\vdash\psi\lollol(\phi\ot(\phi\lol\psi))}$$
		
		6. $$\infrule[$\ot_2$]{\infrule[def $\land$]{\infrule[cut]{\phi\vdash\psi && \infrule[Lemma \ref{l1}.1]{}{\rho\land\phi\vdash\phi}}{\rho\land\phi\vdash\psi}}{\rho\ot(\rho\lol\phi)\vdash\psi}}{\rho\lol\phi\vdash\rho\lol\psi}$$
		
		7. $$\infrule[def $\lor$]{\infrule[$\land_2$]{\infrule[$\ot_2$]{\infrule[id, weak, $\ot_1$]{}{\phi\ot(\psi\lol\phi)\vdash\phi}}{\phi\vdash(\psi\lol\phi)\lol\phi}&&
				\infrule[$\ot_2$]{\infrule[def $\land$]{\infrule[Lemma \ref{l1}.2]{}{\phi\land\psi\vdash\psi}}{\phi\ot(\phi\lol\psi)\vdash\psi}}{\phi\vdash(\phi\lol\psi)\lol\psi}
			}{\phi\vdash((\phi\lol\psi)\lol\psi)\land((\psi\lol\phi)\lol\phi)}}{\phi\vdash\phi\lor\psi}$$
		
		8. Similar to 7.
		
		9. $$\infrule{\infrule[Lemma \ref{l1}.9]{}{\psi\vdash\phi\lor\psi} &&
			\infrule[$\land_1$, def $\lor$]{\infrule{\infrule{\phi\vdash\psi}{\vdash\phi\lol\psi}}{(\phi\lol\psi)\lol\psi\vdash\psi}}{\phi\lor\psi\vdash\psi}}{\vdash \phi\lor\psi\lollol\psi}$$
	\end{proof}
	
	\begin{lemma}[Rules for conjunction and disjunction]\label{disjconj}
    The following are derivable proof rules in $\pl$
		\begin{align*}
			\begin{array}{c}
				\begin{aligned}
					\infrule[$\land_1$]{\Gamma, \phi \vdash \theta}{\Gamma, \phi \land \psi \vdash \theta}
					&&&
					\infrule[$\land_2$]{\Gamma \vdash \phi	& \Gamma \vdash \psi}{\Gamma \vdash \phi \land \psi}
					&&&
					\infrule[$\land_3$]{\Gamma \vdash \phi \land \psi}{\Gamma \vdash \psi}
					\\[2ex]
					\infrule[$\lor_1$]{\Gamma, \phi \vdash \theta & \Gamma, \psi \vdash \theta}{\Gamma, \phi \lor \psi \vdash \theta}
					&&&
					\infrule[$\lor_2$]{\Gamma \vdash \phi}{\Gamma \vdash \phi \lor \psi}
					&&&
					\infrule[$\lor_3$]{\Gamma, \phi \lor \psi \vdash \theta}{\Gamma, \psi \vdash \theta}
				\end{aligned}
			\end{array}
		\end{align*}
	\end{lemma}
	
	\begin{proof} We prove each item separately.
 
		$\land_1$:
		$$\infrule[cut]{\Gamma,\phi\vdash\theta
			&&\infrule[Lemma \ref{l1}.1]{}{\phi\land\psi\vdash\phi}&}
		{\Gamma,\phi\land\psi\vdash\theta}$$
		
		$\land_2$: Assume $\Gamma = \gamma_1, \dots, \gamma_n$ and let $\gamma = \bigoplus_i \gamma_i$. We have
		$$\infrule[$\lol_1$]{\infrule[Lemma \ref{l1}.6]{\infrule[$\ot_1$]{\Gamma\vdash\psi}{\gamma\vdash\psi}}{\phi\lol\gamma\vdash\phi\lol\psi} && \infrule[$\ot_1$]{\Gamma\vdash\phi}{\gamma\vdash\phi}}{\infrule[def $\land$]{\gamma\vdash\phi\ot(\phi\lol\psi)}{\infrule[$\ot_1$]{\gamma\vdash\phi\land\psi}{\Gamma\vdash\phi\land\psi}}}$$
		
		$\land_3$:
		$$\infrule[cut]{\Gamma\vdash\phi\land\psi && \infrule[Lemma \ref{l1}.2]{}{\phi\land\psi\vdash\psi}}{\Gamma\vdash\psi}$$
		
		$\lor_1$: trivial, using Lemma~\ref{totality}.
		
		$\lor_2$:
		$$\infrule[cut]{\Gamma \vdash \phi && \infrule[Lemma \ref{l1}.7]{}{\phi\vdash\phi\lor\psi}}{\Gamma\vdash\phi\lor\psi}$$
		
		$\lor_3$:
		$$\infrule[cut]{\Gamma, \phi\lor\psi\vdash\theta && \infrule[Lemma \ref{l1}.8]{}{\psi\vdash\phi\lor\psi}}{\Gamma,\phi\vdash\theta}$$
	\end{proof}
	
	\begin{lemma}\label{distrib}
        The following are derivable proof rules in $\pl$
		\begin{align*}
			\begin{array}{c}
				\begin{aligned}
					\infrule[1]{}{\vdash\theta(\phi\land\psi)\lollol\theta\phi\land\theta\psi}
					&&&
					\infrule[2]{}{\vdash\theta(\phi\lor\psi)\lollol\theta\phi\lor\theta\psi}
					&&&
					\doubleinfrule[3]{\vdash \phi>0
					}{\vdash\bot\phi\lollol \bot}
				\end{aligned}
				\\
				\begin{aligned}
					\infrule[4]{}{\vdash\min\{r,s\}\phi\lollol r\phi\lor s\phi}
					&&&
					\infrule[5]{}{\vdash\max\{r,s\}\phi\lollol r\phi\land s\phi}
					&&&
					\doubleinfrule[6]{\vdash\phi\psi}{\vdash\phi\lor\psi}
				\end{aligned}
			\end{array}
		\end{align*}
	\end{lemma}
	
	\begin{proof}
		1, 2, 4, 5: trivial using tautology lemma, (\textsc{comp}) and Lemma \ref{l1}.
		
		3:
		\begin{align*} 
			\begin{array}{c}
				\begin{aligned}
					\infrule{\infrule{\infrule{\vdash\phi>0}{\vdash\phi\bot\lol\bot} && \infrule{\infrule{}{\bot\vdash\phi\bot}}{\vdash\bot\lol\phi\bot}}{\vdash(\phi\bot\lol\bot)\land(\bot\lol\phi\bot)}}{\vdash\phi\bot\lollol\bot}
					&&&&
					\infrule{\infrule{\vdash\phi\bot\lollol\bot}{\vdash\phi\bot\lol\bot}}{\vdash\phi>0}	
				\end{aligned}
			\end{array}
		\end{align*}
		
		6: using totality lemma as follows:
		\begin{align*} 
			\begin{array}{c}
				\begin{aligned}
					\infrule{\infrule[cut]{\vdash\phi\psi && \infrule[comp]{\phi\vdash\psi}{\phi\psi\vdash\psi^2}}{\infrule[nullify]{\vdash\psi^2}{\vdash\psi}} && \phi\vdash\psi}{\vdash\phi\lor\psi}
					&&&&
					\infrule{\infrule[cut]{\vdash\phi\psi && \infrule[comp]{\psi\vdash\phi}{\phi\psi\vdash\phi^2}}{\infrule[nullify]{\vdash\phi^2}{\vdash\phi}} && \psi\vdash\phi}{\vdash\phi\lor\psi}
				\end{aligned}
			\end{array}
		\end{align*}
		
		\begin{lemma}\label{order}
        The following are derivable proof rules in $\pl$
			\begin{align*}
				\begin{array}{c}
					\begin{aligned}
						\doubleinfrule[1]{\phi\vdash\psi}{\vdash(\phi>\psi)\lor(\phi=\psi)}
						&&&
						\infrule[2]{}{\vdash(\psi>\psi)\lor(\psi>\phi)\lor(\phi=\psi)}
					\end{aligned}
					\\ \\
					\begin{aligned}
						\infrule[3]{\vdash\phi>\psi&&\vdash\psi>\phi}{\vdash\bot}
						&&&
						\infrule[4]{\vdash\phi>\psi&&\vdash\psi=\phi}{\vdash\bot}
					\end{aligned}
				\end{array}
			\end{align*}
		\end{lemma}
		
	\end{proof}

\end{appendix}

\end{document}